\documentclass[11pt]{article}
\usepackage[a4paper, margin=1in]{geometry}

\usepackage{graphicx}
\graphicspath{ {./images/} }

\usepackage{tcolorbox}

\usepackage{pgfplots}
\usepackage{subcaption}
\pgfplotsset{compat=1.18}

\usepackage{setspace}

\usepackage{booktabs}
\usepackage{arydshln} 
\usepackage{makecell}

\usepackage{algorithm}
\usepackage{algpseudocode}
\usepackage{url}
\usepackage{multirow}

\newtcolorbox{problemBox}{
  colback=white,
  colframe=black,
  boxrule=0.5pt,
  arc=2pt,
  left=4pt,
  right=4pt,
  top=4pt,
  bottom=4pt
}

\usepackage{xcolor, soul}

\usepackage{caption}

\usepackage{enumitem}

\usepackage{authblk}
\usepackage{orcidlink}
\usepackage[numbers]{natbib}

\usepackage{amssymb,amsmath,amsthm}
\usepackage{thmtools,thm-restate}

\newtheorem{definition}{Definition}[section]
\newtheorem{theorem}[definition]{Theorem}
\newtheorem{corollary}[definition]{Corollary}
\newtheorem{lemma}[definition]{Lemma}
\newtheorem{observation}[definition]{Observation}
\newtheorem{example}[definition]{Example}
\newtheorem{open}[definition]{Research Question}
\newtheorem{prop}[definition]{Proposition}
\newtheorem{remark}[definition]{Remark}
\newtheorem{claim}{Claim}

\usepackage{hyperref}



\begin{document}

\title{Stable Matching with Deviators and Conformists\thanks{Sections \ref{sec:pref} and \ref{sec:paramalgo} extend and strengthen the results in a preliminary version of this paper in the Proceedings of AAMAS 2026 \cite{glitznermanloveasaamas26}. The proof of Theorem \ref{thm:preferentialbip} builds on a construction in Augustine Kwanashie's PhD thesis \cite[Thm. 5.3.1]{kwanashie}.}}
\author[1]{Frederik Glitzner}
\author[2]{Augustine Kwanashie\footnote{Work done whilst at the University of Glasgow.}}
\author[1]{David Manlove}
\affil[1]{\small School of Computing Science, University of Glasgow, Glasgow G12 8QQ, UK}
\affil[2]{\small Lloyds Banking Group, Westminster House, 11 Portland Street, Manchester M1 3HU, UK}
\date{}

\maketitle

\begin{abstract}
    In the {\sc Stable Marriage} and {\sc Stable Roommates} problems, there are inherent trade-offs between the size and stability of solutions. While in the former problem, a stable matching always exists and can be found efficiently using the Gale-Shapley algorithm, the existence of a stable matching is not guaranteed in the latter problem, but can be determined efficiently using Irving's algorithm. However, the computation of matchings that minimise the instability, either due to the presence of constraints on the size of the matching or due to restrictive preference cycles, gives rise to a collection of infamously intractable \emph{almost-stable matching} problems.

    In practice, however, not every agent is able or likely to initiate deviations caused by blocking pairs. Suppose we knew, for example, due to a set of requirements or estimates based on historical data, which agents are likely to initiate deviations -- the \emph{deviators} -- and which are likely to comply with whatever matching they are presented with -- the \emph{conformists}. Can we decide efficiently whether a matching exists in which no deviator is blocking, i.e., in which no deviator has an incentive to initiate a deviation? Furthermore, can we find matchings in which only a few deviators are blocking?

    We characterise the computational complexity of this question in bipartite and non-bipartite settings. Surprisingly, these problems prove intractable in very strong ways. On the positive side, we identify polynomial-time and fixed-parameter tractable cases (with respect to the number of deviators and the preference list length), providing novel algorithmics for problems where stability cannot be fully guaranteed, and resolving some open cases from the literature along the way.
\end{abstract}


\clearpage
\section{Introduction}
\label{sec:intro}

The two foundational stable matching problems {\sc Stable Marriage} ({\sc smi}) and {\sc Stable Roommates} ({\sc sri}) were first studied by Gale and Shapley \cite{gale_shapley}. In the former problem, we are given two disjoint sets of agents, traditionally referred to as men and women, with ordinal preferences over each other. In the latter problem, agents are part of just one set and express preferences over other agents in the same set. Agents may express incomplete preferences, deeming some possible partners unacceptable. Now, we seek a \emph{matching} consisting of disjoint pairs of mutually acceptable agents. In practice, not all matchings are equally good -- usually, fairness and robustness requirements apply. For example, and for the purpose of this paper, we mainly consider \emph{stable matchings}, that is, matchings that do not admit any \emph{blocking pairs} consisting of mutually acceptable agents that prefer each other to their current partners (or remain unmatched). This requirement underpins many real-world allocation tasks, such as matching junior doctors (so-called residents) to hospital positions \cite{irving1998matching}, university students to dormitory rooms \cite{abraham_grp}, and network users to communication channels \cite{noma}. Much of the literature focuses on the design of efficient algorithms for centralised matching schemes for such problems \cite{matchup}. A key natural question is the following.

\textbf{Question:} Does there exist a matching in which no agent is blocking?

\citet{gale_shapley} gave an affirmative answer to this question for the {\sc smi} setting and provided an efficient algorithm to find such a stable matching. However, in the {\sc sri} setting, the existence of such a matching is not guaranteed (see \cite{gusfield89} for an extensive structural overview), though \citet{irving_sr} presented an efficient algorithm to decide the existence of a stable matching in this setting. 

However, to make stable matchings useful in practice, two key challenges need to be overcome. First, when allowing agents to express incomplete preferences over agents, there can be a significant trade-off between the stability and the size of a matching, as, even in the {\sc smi} setting, stable matchings might not be maximum-cardinality matchings \cite{biro_sm_10}. Furthermore, as already highlighted above, stable matchings might not exist at all in the {\sc sri} setting \cite{glitzner2025empirics}. The following question arises.

\textbf{Question:} Does there exist a \emph{maximum-cardinality} matching in which no agent is blocking?

Gale and Sotomayor \cite{galesoto85} (for the {\sc smi} setting) and Gusfield and Irving \cite{gusfield89} (for the {\sc sri} setting) observed that the set of matched agents is identical in all stable matchings when preferences are assumed to be strict (an important assumption we make throughout this paper). Consequently, all stable matchings have the same size. However, stable matchings might be significantly smaller than maximum-cardinality matchings, and real-world applications usually require as many agents to be matched as possible. \citet{bertsimas2025relaxstabilitymatchingmarkets}, for example, recently noted that even relaxing the stability requirement by a little bit can give significant benefits with regard to the size of the achievable matchings. One way to do so is to seek a maximum-cardinality matching that is ``as stable as possible''. Eriksson and Häggström \cite{Eriksson2008} reviewed different measures of instability and concluded that, among these measures, the number of blocking pairs is often the most useful measure of instability \cite{matchup}. 

However, are all agents really likely to identify and act on opportunities to deviate in practice? For example, in a school choice or job market setting, are all schools, companies, and applicants equally likely to reach out to better potential partners post-matching to identify blocking pairs? Naturally, the answer in many practical scenarios is not necessarily -- some agents simply do not have the capability, resources, or incentives to detect a beneficial opportunity, while others do. For example, a well-staffed hospital may have the means to follow up on better applicants to check if they have an incentive to come to this hospital, or applicants with more time on their hands may have the means to call up better places they applied to to enquire if they have an incentive to take them. Still, other agents who are unlikely to initiate deviations themselves may still deviate when someone else reaches out to them about a beneficial deviation opportunity. Thus, in such settings, there may be less of a practical need to minimise the instability of these agents. Specifically, in situations where full stability is not achievable, for example due to the presence of additional constraints on the size of the matching, or due to restrictive preference cycles, it may be feasible to simply enforce that no \emph{deviator} (i.e., the agents that are likely to deviate when given a beneficial opportunity) has a beneficial opportunity to deviate, while placing no such requirement on the \emph{conformists} (i.e., the agents that are unlikely to initiate deviations). In this paper, we focus on the following natural question and its extensions across a variety of different settings.

\textbf{Question:} Does there exist a \emph{desirable} matching in which no \emph{deviator} is blocking?

\subsection{Our Contributions}

In this paper, we outline the computational landscape surrounding the computation of (maximum-cardinality) matchings that minimise the instability of a given subset of \emph{deviator agents} in terms of their numbers and sets of blocking pairs. Unfortunately, we derive strong intractability results: even deciding whether there exists a matching (maximum-cardinality matching) in which no deviator agent is in any blocking pair is {\sf NP-complete} for {\sc sri} (respectively {\sc smi}). On the positive side, we present efficient parametrised algorithms (with respect to the number of deviators and the preference list length) and polynomial-time algorithms for instances with short preference lists. We also extend our results to the minimisation of blocking deviator agents and resolve some further open questions from the literature.

\subsection{Structure of the Paper}

In Section \ref{sec:background}, we give formal definitions and discuss relevant related work. Deviator-stability in terms of blocking pairs is introduced and investigated in Section \ref{sec:pref}, followed by parametrised algorithms in Section \ref{sec:paramalgo}. In Section \ref{sec:prefminba}, we extend the results to an alternative blocking agent perspective. Section \ref{sec:prefshort} contains flexible, efficient algorithms for instances with very short lists that can minimise blocking pairs involving deviator agents, or the number of blocking deviators. In Section \ref{sec:furtheropen} we resolve some further open cases from the literature. Finally, we summarise new and known complexity results in Table \ref{table:results} and pose open questions in Section \ref{sec:conclusion}.

\section{Background}
\label{sec:background}

\subsection{Formal Definitions}

In Section \ref{sec:intro}, we introduced the {\sc smi} and {\sc sri} problems informally. Now, we will provide formal definitions of these problem instances that apply throughout this paper.

\begin{definition}[{\sc sri} Instance]
    A {\sc Stable Roommates with Incomplete Lists} ({\sc sri}) \emph{instance}, denoted by $I=(A,\succ)$, consists of a set of $n\in\mathbb N$ agents $A=\{ a_1, a_2, \dots, a_n \}$ and a tuple $\succ$ of $n$ \emph{preference rankings} $\succ_i$ (where $i$ is the index of corresponding agent $a_i$) over a subset of $A\setminus\{a_i\}$. We say that $a_i$ \emph{prefers} $a_j$ to $a_k$ if $a_j\succ_ia_k$. If $a_j\succ_i a_k$ or $j=k$, we will write $a_j\succeq_ia_k$. Two agents $a_i$ and $a_j$ are referred to as \emph{acceptable} if they appear in each other's preference rankings. Throughout, we assume that acceptability is symmetric.
\end{definition}

Note that any {\sc sri} instance $I=(A,\succ)$ can be modelled as a graph $G=(A,E)$, where every agent is a vertex and every acceptable pair of agents induces an edge in $E$. From hereon, $G$ is referred to as the \emph{acceptability graph} of $I$. If $G$ is isomorphic to a complete graph on $\vert A\vert$ vertices, then we say that $I$ has \emph{complete preferences}. Otherwise, we say that $I$ has \emph{incomplete preferences}. The restriction of {\sc sri} instances to instances that have bipartite acceptability graphs is the class of {\sc Stable Marriage with Incomplete Lists} ({\sc smi}) instances. Given a bipartition of $A$ into mutually exclusive but collectively exhaustive subsets $A_1$ and $A_2$, we say that an {\sc smi} instance has complete preferences (with respect to this bipartition) if its acceptability graph is isomorphic to the complete bipartite graph on $\vert A_1\vert,\vert A_2\vert$ vertices.

Throughout, we define matchings and associated concepts as follows.

\begin{definition}[Matchings]
    Let $I=(A,\succ)$ be an {\sc sri} (or {\sc smi}) instance. Then a \emph{matching} $M$ of $I$ is a set consisting of (unordered) disjoint pairs of agents. For any pair $\{a_i,a_j\}\in M$, we refer to $a_j$ as the \emph{partner}, denoted by $M(a_i)$, of $a_i$, and vice versa. If there exists some $a_j$ such that $\{a_i,a_j\}\in M$, then we refer to $a_i$ as \emph{matched} in $M$. Otherwise, we refer to $a_i$ as \emph{unmatched} and denote $M(a_i)=a_i$. 
    
    The set of all matchings of $I$ is given by $\mathcal M$. A \emph{maximum-cardinality matching} is a matching $M$ of $I$ such that, for any other $M'\in \mathcal M$, it is true that $\vert M\vert\geq \vert M'\vert$. The set of all maximum-cardinality matchings is denoted by $\mathcal{M}^+$. Furthermore, maximum-cardinality matchings are \emph{perfect} if every agent is matched (to someone other than themselves). The set of perfect matchings (which might be empty) is denoted by $\mathcal{M}^p$.
\end{definition}

We formalise the classical concept of stability below.

\begin{definition}[Stability]
    Given an {\sc sri} (or {\sc smi}) instance $I=(A,\succ)$ and a matching $M$ of $I$, a pair of distinct agents $\{a_i,a_j\}$ is a \emph{blocking pair} of $M$ if $a_i\succ_j M(a_j)$ and $a_j\succ_i M(a_i)$. The set of blocking pairs of $M$ is denoted by $bp(M)$. Given an agent $a_i$ and a matching $M$, we can restrict the set of blocking pairs to those containing $a_i$, the set of which we denote by $bp_{a_i}(M)$. Finally, we refer to $a_i$ as a \emph{blocking agent} if $bp_{a_i}(M)\neq \varnothing$. The set of all blocking agents of a matching $M$ is denoted by $ba(M)$.
\end{definition}

Any matching that has an empty set of blocking pairs is referred to as \emph{stable}. Any other matching is referred to as \emph{unstable}. Furthermore, we refer to an instance that admits at least one stable matching as \emph{solvable}, and as \emph{unsolvable} otherwise.

We also assume knowledge of the following ``almost-stable'' matching problem that aims to minimise the number of blocking pairs, first formalised by \citet{abraham06}.

\begin{problemBox}
{\sc MinBP-AlmostStable-sri} \\
\textbf{Input:} {\sc sri} instance $I$. \\
\textbf{Output:} A matching $M$ of $I$ such that $\vert bp(M)\vert=\min_{M'\in \mathcal M}\vert bp(M')\vert$.
\end{problemBox}

In the introduction, we highlighted already that even a solvable instance might not admit a stable matching which is simultaneously a maximum-cardinality matching. Given the practical importance of large matchings, another key almost-stable matching problem is to find a maximum-cardinality matching that minimises the number of blocking pairs it admits \cite{biro_sm_10}, which we define below.

\begin{problemBox}
{\sc MinBP-AlmostStable-Max-smi} \\
\textbf{Input:} {\sc smi} instance $I$. \\
\textbf{Output:} A maximum-cardinality matching $M$ of $I$ such that $\vert bp(M)\vert=\min_{M'\in \mathcal M^+}\vert bp(M')\vert$.
\end{problemBox}

To establish key computational complexity results in this paper, we will refer to the decision problem {\sc (2,2)-e3-sat}, which is a restriction of the classical {\sc Satisfiability} problem ({\sc sat}), that is known to be {\sf NP-complete}, see \citet{Berman2003Max3SAT}. The problem is defined as follows.

\begin{problemBox}
{\sc (2,2)-e3-sat} \\[4pt]
\textbf{Input:} Boolean formula $B$ in CNF, where every clause contains exactly three (distinct) literals, and every variable occurs exactly twice unnegated and twice negated. \\[2pt]
\textbf{Question:} Does $B$ admit a satisfying truth assignment?
\end{problemBox}

This rather restrictive definition of classical {\sc sat} has a nice structure. For example, the following formula is a yes-instance: $B=(V_1\lor V_2\lor V_3)\land(\bar{V_1}\lor \bar{V_2}\lor \bar{V_3})\land(V_1\lor \bar{V_2}\lor V_3)\land(\bar{V_1}\lor V_2\lor \bar{V_3})$, as the truth assignment setting $V_1$ to true and $V_2$ and $V_3$ to false satisfies $B$.

\subsection{Related Work}
\label{sec:related}

\citet{abraham06} established that the problem {\sc MinBP-AlmostStable-sri} is {\sf NP-hard} even for complete preference lists, and \citet{biro12} extended this intractability result to instances where every agent has a preference list of length at most 3. The first set of authors also proved that it is {\sf NP-hard} to approximate {\sc MinBP-AlmostStable-sri} within any $n^{1/2-\varepsilon}$ (for any $\varepsilon>0$, where $n$ is the number of agents), and \citet{chen17} proved that the problem is {\sf W[1]-hard} with respect to the parameter $\min_{M\in \mathcal M}\vert bp(M)\vert$ even when preference lists are of length at most 5.

For {\sc smi}, Biró, Manlove and Mittal \cite{biro_sm_10} studied the trade-off between matching size and stability from a computational complexity perspective. They showed that {\sc MinBP-AlmostStable-Max-smi} (and a similar problem minimising the number of blocking agents) is {\sf NP-hard} even to approximate within $n^{1-\varepsilon}$ (for any $\varepsilon>0$) \cite{biro_sm_10}. Later, Hamada et al. \cite{hamada09} strengthened this result to hold even in the presence of short preference lists of length at most 3. In contrast, for preference lists of length at most 2, Biró, Manlove and Mittal \cite{biro_sm_10} showed that {\sc MinBP-AlmostStable-Max-smi} (and the minimisation of blocking agents) is polynomial-time solvable. However, {\sc MinBP-AlmostStable-Max-smi} is {\sf W[1]-hard} even for a variety of combined parameters, which was studied extensively by Gupta et al. \cite{gupta2020parameterized}. Recently, Chen, Roy and Simola \cite{chen2025fptapproximabilitystablematchingproblems} proved in a preprint that it is {\sf W[1]-hard} (with respect to the optimal value) even to approximate the problems {\sc MinBP-AlmostStable-sri} and {\sc MinBP-AlmostStable-Max-smi} within any computable function depending only on the optimal solution value.

Complementary to the optimisation problems, we denote the decision problems asking whether there exists a matching (respectively a perfect matching in the {\sc smi} setting) with at most $k$ blocking pairs by {\sc $k$-BP-AlmostStable-sri} and {\sc $k$-BP-AlmostStable-Perfect-smi}. It is known that both of these problems are {\sf NP-complete} in general, but solvable in polynomial time for a fixed parameter $k$ (i.e., in {\sf XP} with respect to the fixed parameter \cite{flumgrohe}) \cite{abraham06,biro12,biro_sm_10,chen2019computational}. 

\citet{glitznermanloveminmax} recently studied an alternative minimax notion of almost-stability, where the maximum number of blocking pairs that any agent is contained in is minimised. They proved, for example, that even deciding whether there exists a matching in which no agent is in more than one blocking pair is {\sf NP-complete}. On the positive side, they provided efficient algorithms for instances with preference lists of length at most 3, and an approximation algorithm for the {\sc sri} setting.

We note that Floréen et al. \cite{Floreen2010} and Ostrovsky and Rosenbaum \cite{OstrovskyRosenbaum2015} considered other definitions of almost-stability, and almost-stable matchings have also been studied in other contexts, such as the {\sc Hospitals/Residents} problem with couples \cite{minbp_hrc_17,couples24}. Furthermore, there is a wide literature on other solution concepts for {\sc sri} instances \cite{gupta_popsr_21, faenza_pop,tan91_2,herings25,vandomme2025locally,glitzner2024structural,glitzner2025unsolvability}.

The assumption that underlies our problem models (allowing some agents to block with each other, while disallowing other agents to block) is also related to so-called \emph{socially-stable matchings} and the related concept of \emph{free edges}.

In the socially-stable matching problem (in the interpretation of \cite{sociallystable}), there is an {\sc smi} instance (or a many-to-one {\sc Hospitals/Residents} instance), as well as a social network graph where the agents are the vertices and edges between agents indicate the social connections, and the aim is to find the largest \emph{socially-stable} matching, that is, the largest matching that admits no \emph{socially-blocking pair} (a blocking pair of agents who are connected by an edge in the social network graph). This problem is denoted by {\sc max-smiss}. In our context, social stability and deviator stability coincide when the social network graph consists of every pair of (acceptable) agents that contains at least one deviator. {\sc max-smiss} is known to be {\sf NP-hard} even for preference lists of length at most 3 and {\sf NP-hard} to approximate within some constant \cite{sociallystable}. However, this {\sf NP-hardness} does not directly carry over to our problems, as social stability can model strictly more general blocking behaviours than our deviator stability notion.

Free edges are the direct complement of socially-blocking pairs. They are pairs of agents that are never considered to be blocking. All acceptable pairs of agents not connected in the social network graph can be considered free edges. The {\sc sri} problem with free edges was first studied by \citet{cechlarovafree}, who proved that deciding the existence of a stable matching in the presence of free edges is {\sf NP-complete}. In our context, stability with free edges and deviator stability coincide exactly when the free edges are all edges between (acceptable) conformists.

Finally, we note the (absence of a) connection between our concepts and \emph{forced and forbidden edges}, and \emph{critical edges and critical agents}. Forced edges are pairs of agents that must be contained in the matching, which forces a certain structure on the solutions, and forbidden edges are pairs of agents that cannot be part of the matching but can block the matching \cite{dias03,cseh16,cseh20}. Neither of these notions is directly related to our model. Critical matchings match a maximum number of critical vertices, that is, vertices from a given set that should be matched as much as possible. It is easy to see that there are settings where no critical matching is stable, so a relaxed stability notion has been studied where a pair of agents that prefer each other to their partners (who may be none) is not considered blocking (despite satisfying the standard definition of a blocking pair) as long as at least one of these agents is matched to a critical vertex \cite{Krishnaa2022,nasre23}. The motivation behind this relaxed stability definition is that since maximising matched critical agents is the primary objective, we do not allow blocking pairs where one of the agents would leave a critical agent unmatched. The notion does not directly capture our model, because the motivations are different: we want to minimise the instability involving the deviator agents, possibly while maximising the overall matching size (without prioritising the number of matched deviators specifically).

Recently, \citet{csaji2025simple} unified the perspectives on free edges and critical relaxed stability, with the central result being a $3/2$-approximation algorithm for a generalised version of {\sc max-smiss} to free edges, critical vertices, and cardinal (rather than ordinal) preferences, and matroid constraints. In the same paper, the authors also introduced the concept of \emph{free agents}, who cannot participate in any blocking pair (i.e., blocking pairs involving free agents do not cause instability). At first sight, free agents may appear to be equivalent to conformists, but this is not true, because blocking pairs involving one conformist and one deviator are possible, while one free agent and one non-free agent cannot form a blocking pair.

\section{Deciding Deviator-Stability is Computationally Intractable}
\label{sec:pref}

It is clear from previous work that minimising either the aggregate or the individual instability among the agents is generally intractable in a variety of critical real-world scenarios. However, generally, not every agent is able to initiate a deviation even when a beneficial opportunity to do so presents itself, for example, when they are prevented from exploring or initiating possible deviations by some external constraints. As in the introduction, we will refer to these agents as \emph{conformists}, and to the agents that are likely to deviate as \emph{deviators}. When the number of conformists is large, naturally, it should be easier to find a desirable matching in which the number of blocking deviators is small. In practical settings where full stability is not achievable, this might be a second-best solution. Hence, we will now introduce optimisation and decision problems that aim for matchings that give a form of preferential treatment to the deviator agents, or more generally, any designated subset of agents. Specifically, we seek matchings that minimise the number of blocking pairs involving these agents, or ideally, exclude them from any blocking pairs altogether.

We first propose the following natural optimisation problem, which seeks a matching that minimises the number of blocking pairs involving a given set of deviator agents $D$. Intuitively, the problem asks: how stable can a matching be made for a privileged subset of agents?

\begin{problemBox}
{\sc Deviator-sri} \\
\textbf{Input:} {\sc sri} instance $I=(A,\succ)$ and a subset $D\subseteq A$. \\
\textbf{Output:} A matching $M\in\mathcal M$ such that $\vert\bigcup_{a_i\in D}bp_{a_i}(M)\vert=\min_{M'\in \mathcal M}\vert\bigcup_{a_i\in D}bp_{a_i}(M')\vert$.
\end{problemBox}

We will also study the maximum-cardinality matching analogue, which we define formally below.

\begin{problemBox}
{\sc Deviator-Max-sri} \\
\textbf{Input:} {\sc sri} instance $I=(A,\succ)$ and a subset $D\subseteq A$. \\
\textbf{Output:} A matching $M\in\mathcal {M}^+$ such that $\vert\bigcup_{a_i\in D}bp_{a_i}(M)\vert=\min_{M'\in \mathcal{M}^+}\vert\bigcup_{a_i\in D}bp_{a_i}(M')\vert$.
\end{problemBox}

In search of more computational tractability, we will also consider the above problem in the bipartite {\sc smi} restriction, where we are given an {\sc smi} instance $I$ instead, and a subset $D\subseteq A$. We will refer to this problem as {\sc Deviator-Max-smi}. Of course, all three problems are {\sf NP-hard} in general, because the {\sf NP-hard} problems {\sc MinBP-AlmostStable-sri} and {\sc MinBP-AlmostStable-Max-smi} appear as special cases when $D=A$. However, we will investigate natural restrictions and algorithms for the case when $D$ is small, and we will also show that, contrary to their non-deviator analogues, these problems turn out to be intractable even when the optimal value is 0, i.e., when merely aiming to decide whether a deviator-stable matching exists! This contrasts known complexity results as follows: deciding whether the optimal value is 0 is efficiently solvable using the Gale-Shapley or Irving's algorithm when $D=A$, and {\sc MinBP-AlmostStable-sri} is in {\sf XP} with respect to the optimal parameter value, i.e., efficiently solvable for any fixed optimal parameter value \cite{abraham06}.

We will first examine maximum-cardinality matchings with these properties in bipartite acceptability graphs (i.e., the {\sc Deviator-Max-smi} problem) in Section \ref{sec:prefsmi}, before extending our techniques to the general case in Section \ref{sec:prefsri}.

\subsection{Deviator-Stable Maximum Matchings in Bipartite Graphs}
\label{sec:prefsmi}

We start by defining the following restriction of {\sc Deviator-Max-smi}.

\begin{problemBox}
{\sc $k$-Deviator-Max-smi} \\
\textbf{Input:} {\sc smi} instance $I=(A,\succ)$ and a subset $D\subseteq A$. \\
\textbf{Output:} A matching $M\in\mathcal {M}^+$ such that $\vert\bigcup_{a_i\in D}bp_{a_i}(M)\vert\leq k$, if one exists.
\end{problemBox}

We will denote the problem of deciding whether there exists a solution to the above problem in the special case when maximum-cardinality matchings are perfect by {\sc $k$-Deviator-Perfect-smi-Dec}. 

Notice that whenever $D=A$ and $k=0$, then $(I,D)$ is a yes-instance to {\sc 0-Deviator-Max-smi} if and only if there exists a stable maximum-cardinality matching of $I$. This can be decided in polynomial time by computing a stable matching $M$ of $I$ using the Gale-Shapley algorithm \cite{gale_shapley} and checking whether $M$ is a maximum-cardinality matching. If it is not, then no solution exists, as every stable matching is of equal size \cite{alroth86}.

Somewhat surprisingly, when $D\subset A$, then even {\sc 0-Deviator-Perfect-smi-Dec} becomes intractable, i.e., even when we merely ask whether a perfect matching exists such that no agent in $D$ is in a blocking pair.

\begin{theorem}
\label{thm:preferentialbip}
    {\sc 0-Deviator-Perfect-smi-Dec} is {\sf NP-complete}, even if all preference lists are of length at most 3 and all deviators are on one side of the bipartition.
\end{theorem}
\begin{proof}
    Membership in {\sf NP} follows from the fact that, given a matching $M$, we can efficiently check whether $M$ is a perfect matching and whether any agent in the deviator subset is contained in any blocking pair by iterating through the preference lists.

    Our reduction is similar to that of \citet{biro_sm_10} from {\sc (2,2)-e3-sat} to {\sc $k$-BP-AlmostStable-Perfect-smi}, but the preference lists of the $p_j^s$-agents and our overall analysis differ.
    
    Formally, consider a {\sc (2,2)-e3-sat} instance $B$ with $n$ variables $V=\{V_1,V_2,\dots,V_n\}$ and $m$ clauses $C=\{C_1,C_2,\dots,C_m\}$. Below, we show how to derive an {\sc smi} instance $I=(A,\succ)$ from $V$ and $C$. In the construction, $A$ will consist of subsets $A_v=\{x_i^r,y_i^r\;;\;1\leq i\leq n\land 1\leq r\leq 4\}$ (from the variable gadgets) and $A_c=\{c_j^s,p_j^s\;;\;1\leq j\leq m\land 1\leq s\leq 3\}\cup\{q_j,z_j\;;\;1\leq j\leq m\}$ (from the clause gadgets). 
    The total number of agents of $I$ is $8(n+m)$.

    We will now show how to construct $I$ in full detail. First, for each variable $V_i\in V$, we construct a \textbf{variable gadget} $G_{V_i}$ consisting of the eight agents $x_i^1,x_i^2,x_i^3,x_i^4,y_i^1,y_i^2,y_i^3,y_i^4$, with the following preference lists: 
    \[
    \begin{minipage}{0.45\textwidth}
    \begin{align*}
    x_i^1 &: y_i^1 \; c(x_i^1) \; y_i^2 \\
    x_i^2 &: y_i^2 \; c(x_i^2) \; y_i^3 \\
    x_i^3 &: y_i^4 \; c(x_i^3) \; y_i^3 \\
    x_i^4 &: y_i^1 \; c(x_i^4) \; y_i^4
    \end{align*}
    \end{minipage}
    \hfill
    \begin{minipage}{0.45\textwidth}
    \begin{align*}
    y_i^1 &: x_i^1 \; x_i^4 \\
    y_i^2 &: x_i^1 \; x_i^2 \\
    y_i^3 &: x_i^2 \; x_i^3 \\
    y_i^4 &: x_i^3 \; x_i^4
    \end{align*}
    \end{minipage}
    \]
    where we will specify the $c(x_i^r)$ entries later. Figure \ref{fig:vsmigadget} illustrates this construction. Importantly, a $G_{V_i}$ gadget admits exactly two perfect matchings: the matchings 
    \begin{align*}
        M^1_{V_i}&=\{\{x_i^1,y_i^1\},\{x_i^2,y_i^2\},\{x_i^3,y_i^3\},\{x_i^4,y_i^4\}\}\text{ and}\\
        M^2_{V_i}&=\{\{x_i^1,y_i^2\},\{x_i^2,y_i^3\},\{x_i^3,y_i^4\},\{x_i^4,y_i^1\}\}.
    \end{align*}
    Notice that $bp(M^1_{V_i})\cap \mathbb{P}_2(A_v)=\{\{x_i^3,y_i^4\}\}$ and $bp(M^2_{V_i})\cap \mathbb{P}_2(A_v)=\{\{x_i^1,y_i^1\}\}$, where $\mathbb{P}_2(A_v)$ is the set of unordered pairs of agents in $A_v$, i.e., we are not considering any blocking pairs that involve $c(x_i^r)$ agents. Intuitively, $M^1_{V_i}$ and $M^2_{V_i}$ correspond to $V_i$ being true or false, respectively.

    \begin{figure}[!tbh]
        \centering
        \includegraphics[width=7cm]{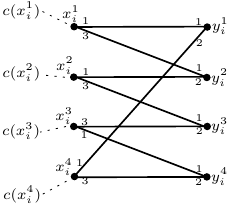}
        \caption{Illustration of the variable gadget construction}
        \label{fig:vsmigadget}
    \end{figure}

    Now, for each clause $C_j=(C_j^1\lor C_j^2 \lor C_j^3)\in C$, we construct a corresponding \textbf{clause gadget} $G_{C_j}$. A clause gadget consists of eight agents $c_j^1, c_j^2, c_j^3, p_j^1, p_j^2, p_j^3, q_j, z_j$, with the following preference lists: 
    \[
    \begin{minipage}{0.45\textwidth}
    \begin{align*}
    c_{j}^1 &: p_{j}^1 \; x(c_{j}^1) \; q_{j} \\
    c_{j}^2 &: p_{j}^2 \; x(c_{j}^2) \; q_{j} \\
    c_{j}^3 &: p_{j}^3 \; x(c_{j}^3) \; q_{j} \\
    z_{j} &: p_{j}^1 \; p_{j}^2 \; p_{j}^3
    \end{align*}
    \end{minipage}
    \hfill
    \begin{minipage}{0.45\textwidth}
    \begin{align*}
    p_{j}^1 &: z_{j} \; c_{j}^1 \\
    p_{j}^2 &: z_{j} \; c_{j}^2\\
    p_{j}^3 &: z_{j} \; c_{j}^3\\
    q_{j} &: c_{j}^1 \; c_{j}^2 \; c_{j}^3
    \end{align*}
    \end{minipage}
    \]
    where we will specify the $x(c_j^s)$ entries later. Figure \ref{fig:csmigadget} illustrates this construction. Importantly, a $G_{C_j}$ gadget admits the following three perfect matchings: 
    \begin{align*}
        M^1_{C_j}&=\{\{c_{j}^1,q_{j}\},\{c_{j}^2,p_{j}^2\},\{c_{j}^3,p_{j}^3\},\{p_{j}^1,z_{j}\}\},\\
        M^2_{C_j}&=\{\{c_{j}^1,p_{j}^1\},\{c_{j}^2,q_{j}\},\{c_{j}^3,p_{j}^3\},\{p_{j}^2,z_{j}\}\},\text{ and}\\
        M^3_{C_j}&=\{\{c_{j}^1,p_{j}^1\},\{c_{j}^2,p_{j}^2\},\{c_{j}^3,q_{j}\},\{p_{j}^3,z_{j}\}\}.
    \end{align*}
    Notice that $bp(M^1_{C_j})\cap \mathbb{P}_2(A_c)=\varnothing$, $bp(M^2_{C_j})\cap \mathbb{P}_2(A_c)=\{\{p_j^s,z_j\}\}$, and $bp(M^3_{C_j})\cap \mathbb{P}_2(A_c)=\{\{p_{j}^1,z_j\},\{p_{j}^2,z_j\}\}$. Intuitively, $M^1_{C_j},M^2_{C_j}$ and $M^3_{C_j}$ correspond to the first, second, and third literals of $C_j$ being true, respectively.

    \begin{figure}[!tbh]
        \centering
        \includegraphics[width=8cm]{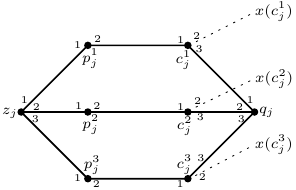}
        \caption{Illustration of the clause gadget construction}
        \label{fig:csmigadget}
    \end{figure}

    We now connect the gadgets in the following way: let $C_j^s$ correspond to the $s$th literal in $C_j$ (for $1\leq j\leq m, 1\leq s\leq 3$), which is either an unnegated or a negated occurrence of variable $V_i$. Assuming that $B$ is read sequentially, if $C_j^s$ is the
    \begin{itemize}
        \item first unnegated occurrence of $V_i$, then let $x(c_j^s)=x_i^1$ and let $c(x_i^1)=c_j^s$;
        \item second unnegated occurrence of $V_i$, then let $x(c_j^s)=x_i^2$ and let $c(x_i^2)=c_j^s$;
        \item first negated occurrence of $V_i$, then let $x(c_j^s)=x_i^3$ and let $c(x_i^3)=c_j^s$;
        \item second negated occurrence of $V_i$, then let $x(c_j^s)=x_i^4$ and let $c(x_i^4)=c_j^s$.
    \end{itemize}

    Now we define the set of deviators to be the $c_j^s$ agents, i.e., $D=\{c_j^s\;;\; 1\leq j\leq m\land 1\leq s\leq 3\}$, and prove the following claim.
    
    \begin{claim}
        \label{claim:psmi}
        $B$ is satisfiable if and only if $I$ admits a perfect matching $M$ such that no agent in $D$ is in a blocking pair.
    \end{claim}
    \begin{proof}[Proof of Claim \ref{claim:psmi}]
    \renewcommand{\qedsymbol}{$\blacksquare$}
        First, suppose that $B$ is satisfiable and that $f$ is a satisfying truth assignment. We define a complete matching $M$ in $I$ as follows. For each variable $V_i\in V$, if $f(V_i)$ is true, then add $M_{V_i}^1$ to $M$, otherwise add $M_{V_i}^2$ to $M$. For each clause $C_j\in C$, there must exist at least one literal that is true under $f$ by assumption. Let this be $C_j^s$ (if multiple literals satisfy $C_j$, pick an arbitrary true literal of $C_j$). Then we add $M_{C_j}^s$ to $M$. Clearly, at the conclusion of this process, $M$ is a perfect matching, as we include one perfect matching within each gadget in $M$. Now it is straightforward to see that no agent in $D$ is blocking. We already noted that there can be no blocking pair involving an agent in $D$ within the clause gadgets. Now consider any pair $\{c_j^s,x(c_j^s)\}$ for $1\leq j\leq m $ and $1\leq s\leq 3$. If $C_j^s$ is the chosen satisfying literal in $C_j$, then $x(c_j^s)$ is matched to their first choice in $M$, so they cannot block with $c_j^s$. Otherwise, if $C_j^s$ is not the chosen satisfying literal in $C_j$, then $c_j^s$ is matched to their first choice ($p_j^s$) in $M$, so they cannot block with $x(c_j^s)$.

        For the other direction, suppose that $M$ is a perfect matching in $I$ such that no agent in $D$ is in a blocking pair. As $M$ is perfect, either $M_{V_i}^1\subseteq M$ or $M_{V_i}^2\subseteq M$ for every variable gadget $G_{V_i}$ ($1\leq i\leq n$), and either $M_{C_j}^1\subseteq M$, $M_{C_j}^2\subseteq M$ or $M_{C_j}^3\subseteq M$ for the clause gadget $G_{C_j}$ ($1\leq j\leq m$). We define a truth assignment $f$ of $B$ as follows. For each variable $V_i\in V$, if $M_i^1\subseteq M$ then set $f(V_i)$ to true, otherwise (it must be true that $M_{V_i}^2\subseteq M$) set $f(V_i)$ to false. Naturally, this yields a valid truth assignment. We now show that it satisfies $B$. Suppose for the sake of contradiction that there exists a clause $C_j\in C$ such that no literal in $C_j$ is satisfied. Consider the matching $M_{C_j}^s\subseteq M$ corresponding to the clause gadget $G_{C_j}$. By construction, it must be the case that $\{c_j^s,q_j\}\in M_j^s$ and agent $c_j^s$ prefers $x(c_j^s)$ to their partner in $M$. However, because $c_j^s\in D$ and we assume that no agent in $D$ is in a blocking pair, we know that $x(c_j^s)$ has a partner better than $c_j^s$. Hence, by construction, either $C_j^s$ is an unnegated literal, and we have already set the corresponding variable to true, or $C_j^s$ is a negated literal and we have already set the corresponding variable to false. In both cases, $f$ satisfies $C_j^s$, which is a contradiction.
    \end{proof}

    Our claim above immediately establishes that {\sc 0-Deviator-Perfect-smi-Dec} is {\sf NP-hard}. Notice that every agent in our construction has a preference list of length at most 3, and that all agents in $D$ are on the same side of the bipartition where the $c,y$ and $z$ agents are on one side, and the $p,q$ and $x$ agents are on the other side. This also completes the proof that {\sc 0-Deviator-Perfect-smi-Dec} is {\sf NP-complete}.
\end{proof}

We can immediately conclude the following intractability result for our original optimisation problem.

\begin{corollary}
\label{cor:devmaxsmihard}
    {\sc $k$-Deviator-Max-smi} is {\sf para-NP-hard} with respect to $k$ and, in particular, is {\sf NP-complete} even when $k=0$ and all preference lists are of length at most 3 and all deviators are on the same side. Thus, {\sc Deviator-Max-smi} is {\sf NP-hard} in equally restricted settings, and cannot be in {\sf XP} with respect to the optimum parameter value, unless {\sf P}={\sf NP}.
\end{corollary}

Note that these hardness results no longer apply when preference lists are complete, since then every stable matching is a maximum-cardinality matching \cite{alroth86}.

\begin{remark}
    \label{rm:otherreduction}
    A reviewer of a preliminary version of this paper pointed out that a new analysis of the reduction in the proof of Theorem 3 of \cite{gupta2020parameterized} can also establish that {\sc $0$-Deviator-Max-smi} (which we newly introduced in this paper) is {\sf NP-hard}, although not for preference lists of length at most 3. 
    However, given that we establish in Section \ref{sec:prefshort} that {\sc $k$-Deviator-Max-smi} is in {\sf P} for instances with lists of length at most 2, there is direct merit in our reduction, as it establishes a tight complexity frontier with respect to the preference list length. 
\end{remark}

\subsection{Deviator-Stable Matchings in General Acceptability Graphs}
\label{sec:prefsri}

We now turn to the general case, where matchings may have arbitrary size, and the acceptability graph need not be bipartite. We start by defining a restricted version of {\sc Deviator-sri} as follows.

\begin{problemBox}
{\sc $k$-Deviator-sri} \\
\textbf{Input:} {\sc sri} instance $I=(A,\succ)$ and a subset $D\subseteq A$. \\
\textbf{Output:} A matching $M\in\mathcal M$ such that $\vert\bigcup_{a_i\in D}bp_{a_i}(M)\vert\leq k$, if one exists.
\end{problemBox}

Again, we will focus on the very restricted case of $k=0$, where we require a matching in which no agent in $D$ is in a blocking pair, and denote the problem of deciding whether {\sc 0-Deviator-sri} admits a solution by {\sc 0-Deviator-sri-Dec}. 

We start with the following observation about a special class of {\sc 0-Deviator-sri} instances.

\begin{theorem}
\label{thm:0deviatorsri}
    Let $(I=(A,\succ),D)$ be an instance with $n$ agents of {\sc 0-Deviator-sri} and let $G=(A,E)$ be the acceptability graph of $I$. If $G'=(A,E\setminus E[A\setminus D])$ is bipartite, then we can solve {\sc 0-Deviator-sri} in $O(n^2)$ time.
\end{theorem}
\begin{proof}
    Suppose that $G'=(A, E\setminus E[A\setminus D])$ is bipartite. Then the instance $I'=(A,\succ')$ induced by the restricted acceptability graph $G'$ is an {\sc smi} instance. Any {\sc smi} instance admits a stable matching $M$, and we can find one in $O(n^2)$ time  \cite{gale_shapley}. Now suppose that $M$ admits a blocking pair $\{a_r,a_s\}\in bp(M)$ with respect to instance $I$. Then, by construction of $I'$, it must be the case that $\{a_r,a_s\}\in E[A\setminus D]$, in which case $\bigcup_{a_i\in D}bp_{a_i}(M)=\varnothing$ as required.
\end{proof}

\begin{figure}[!tbh]
    \centering
    \includegraphics[width=7cm]{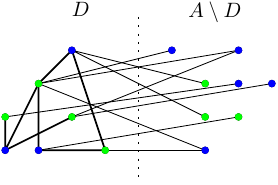}
    \caption{Illustration of restricted acceptability graph (colours indicate a possible bipartition)}
    \label{fig:restricted}
\end{figure}

An example of a restricted preference graph can be seen in Figure \ref{fig:restricted}. Unfortunately, this structural insight does not generalise to an efficiently verifiable necessary and sufficient condition for {\sc 0-Deviator-sri-Dec}, as there may be complex interactions among agents between and outside $D$. The following intractability result leaves little hope for an efficiently verifiable condition for the existence of a solution to {\sc 0-Deviator-sri}.

\begin{theorem}
\label{thm:prefsrihard}
    {\sc 0-Deviator-sri-Dec} is {\sf NP-complete}, regardless of whether all preference lists are of length at most 5 or all preference lists are complete.
\end{theorem}
\begin{proof}
    We will reduce from {\sc 0-Deviator-Perfect-smi-Dec}, which we showed to be {\sf NP-complete} in Theorem \ref{thm:preferentialbip}. Notice that membership in {\sf NP} follows immediately, as we can simply iterate through $D$ and check that no agent in this set is contained in any blocking pairs of a given perfect matching.

    Let $(I=(A,\succ),D)$ be an instance of {\sc 0-Deviator-Perfect-smi-Dec} and let $(J=(A_J,\succ^J),D_J)$ be the instance of {\sc 0-Deviator-sri-Dec} we construct, where $A_J=A\cup A_B$, $A_B=\{b_i^1,b_i^2\;;\;1 \leq i\leq \vert A\vert \}$, and $D_J=D\cup A_B$. Overall, $J$ has three times the number of agents as $I$. Now, for every agent $a_i\in A$ with preference list $P_{a_i}$ in $I$, we construct the following preference lists:
    \begin{align*}
        a_i &: P_{a_i} \; b_i^1 \; b_i^2 \\
        b_i^1 &: b_i^2 \; a_i \\
        b_i^2 &: a_i \; b_i^1
    \end{align*}

    We claim the following correspondence between the two problems. For clarity, we will indicate the relevant instance in the notation of blocking pairs, e.g., $bp^I(M)$ denotes the set of blocking pairs that $M$ admits with respect to instance $I$ and $bp_{a_i}^I(M)$ denotes the set of blocking pairs involving $a_i$ that $M$ admits with respect to instance $I$.

    \begin{claim}
    \label{claim:psri}
        $I$ admits a perfect matching $M$ such that $\bigcup_{a_i\in D}bp_{a_i}^I(M)=\varnothing$ if and only if $J$ admits a matching $M'$ such that $\bigcup_{a_i\in D_J}bp_{a_i}^J(M')=\varnothing$.
    \end{claim}
        \begin{proof}[Proof of Claim \ref{claim:psri}]
        \renewcommand{\qedsymbol}{$\blacksquare$}
        Suppose that $I$ admits a perfect matching $M$ such that $\bigcup_{a_i\in D}bp_{a_i}^I(M)=\varnothing$. Define $M'=M\cup \{\{b_i^1,b_i^2\} \;;\; 1\leq i\leq \vert A\vert\}$. Then $M'$ is a perfect matching of $J$. Furthermore, for every $a_i\in A$, by construction of $J$, we have that $M(a_i)=M'(a_i)\succ_{a_i}^Jb_i^{\gamma}$ (for $\gamma\in \{1,2\}$). Hence, $bp_{a_i}^J(M')=bp_{a_i}^I(M)$ and for all $b_q^\gamma\in A_B$, $bp_{b_q^\gamma}(M')=\varnothing$. Therefore, $\bigcup_{a_i\in D_J}bp_{a_i}^J(M')=\bigcup_{a_i\in D\cup A_B}bp_{a_i}^J(M')=\varnothing$ as required.

        Conversely, suppose that $J$ admits a matching $M'$ such that $\bigcup_{a_i\in D_J}bp_{a_i}^J(M')=\varnothing$ and consider the matching $M=M'\setminus \{\{b_i^1,b_i^2\} \;;\; 1\leq i\leq \vert A\vert\}$. Then for all $a_i\in A$, it must be true that $M'(a_i)\succ_{a_i}^J b_i^{\gamma}$ (for $\gamma\in\{1,2\}$), otherwise at least one of $\{a_i,b_i^1\}$, $\{a_i,b_i^2\}$ or $\{b_i^1,b_i^2\}$ would block $M'$ in $J$, in which case $\bigcup_{a_i\in A_B}bp_{a_i}^J(M')\neq\varnothing$, a contradiction. Hence, $M$ is a perfect matching of $I$. Furthermore, clearly $\bigcup_{a_i\in D}bp_{a_i}^I(M)=\bigcup_{a_i\in D}bp_{a_i}^J(M')\subseteq \bigcup_{a_i\in D_J}bp_{a_i}^J(M')=\varnothing$ as required.
    \end{proof}

    Due to this correspondence, and by the fact that {\sc 0-Deviator-Perfect-smi-Dec} is {\sf NP-complete} even when preference lists are of length at most 3, we can immediately conclude that {\sc 0-Deviator-sri-Dec} is {\sf NP-complete} even when preference lists are of length at most 5. 

    Now, to extend this result to complete preference lists, we construct the instance $J'=(A_{J},\succ^{J'})$ from $J$ by arbitrarily appending previously unranked agents to the end of the preference lists of all agents. To see that the intractability result holds even when preference lists are complete, observe the following.

    \begin{claim}
    \label{claim:psrc}
        $J$ admits a matching $M$ such that $\bigcup_{a_i\in D_J}bp_{a_i}(M)=\varnothing$ if and only if $J'$ admits a matching $M'$ such that $\bigcup_{a_i\in D_J}bp_{a_i}(M')=\varnothing$.
    \end{claim}
    \begin{proof}[Proof of Claim \ref{claim:psrc}]
    \renewcommand{\qedsymbol}{$\blacksquare$}
        Suppose that $J$ admits a matching $M$ such that $\bigcup_{a_i\in D_J}bp_{a_i}^J(M)=\varnothing$. We already characterised the set of such matchings in the proof of Claim \ref{claim:psri} as complete matchings in $J$. Clearly $M$ also satisfies $\bigcup_{a_i\in D_J}bp_{a_i}^{J'}(M)=\varnothing$.

        Now suppose conversely that $J'$ admits a matching $M'$ such that $\bigcup_{a_i\in D_J}bp_{a_i}^{J'}(M')=\varnothing$. It remains to show that $M'$ only contains pairs that are acceptable in $J$, in which case also $\bigcup_{a_i\in D_J}bp_{a_i}^{J}(M')=\varnothing$. Clearly, it remains true that any agent $a_i\in A$ is matched to someone on $P_{a_i}$ (by the same argument as in Claim \ref{claim:psri}. Furthermore, suppose that some agent $b_i^\gamma$ is not matched to $b_i^{\gamma+1}$ ($\gamma\in\{1,2\}$ and addition taken modulo 2) in $M'$. Then, because previously unranked agents were appended to the end of the preference lists, it must be the case that $b_i^\gamma\succ_{b_i^{\gamma+1}}^{J'} M(b_i^{\gamma+1})$ and $b_i^{\gamma+1}\succ_{b_i^{\gamma}}^{J'} M(b_i^\gamma)$. Hence, $\{b_i^1,b_i^2\}\in bp^{J'}(M')$, which contradicts our assumption $\bigcup_{a_i\in D_J}bp_{a_i}^{J'}(M')=\varnothing$ as $\{b_i^1,b_i^2\}\subseteq D_J$. Thus, $\{b_i^1,b_i^2\}\in M'$, so every match in $M'$ is acceptable with respect to $J$ too.
    \end{proof}
    
    This completes the proof of Theorem \ref{thm:prefsrihard}.
\end{proof}

\begin{remark}
\label{rm:freeedges}
    We note that this intractability result could also be derived from a different {\sc 3-sat} reduction in the proof of Theorem 2.1 of technical report \cite{cechlarovafree} on the {\sc sri} with free edges problem.\footnote{Theorem 2.1 in \cite{cechlarovafree} says that deciding the existence of a stable matching with free edges, where the set of free edges is spanned by a subset of the vertices, is {\sf NP-complete}. When taking the set of edges between (acceptable) conformists, i.e., the edges spanned by the set of conformists, as the free edges, the results become equivalent.} However, we include our method for completeness, as it is a simple extension of our reduction in Theorem \ref{thm:preferentialbip}.
\end{remark}

We can conclude the following intractability result for our original optimisation problem of interest using a similar argument as in the proof of Corollary \ref{cor:devmaxsmihard}.

\begin{corollary}
\label{cor:devsrihard}
    {\sc $k$-Deviator-sri} is {\sf para-NP-hard} with respect to $k$ and, in particular, is {\sf NP-complete} even when $k=0$ and regardless of whether preference lists are of length at most 5 or complete. Thus, {\sc Deviator-sri} is {\sf NP-hard} in equally restricted settings, and cannot be in {\sf XP} with respect to the optimum parameter value, unless {\sf P}={\sf NP}.
\end{corollary}

We reiterate that {\sc MinBP-AlmostStable-sri} is in {\sf XP} with respect to the optimal parameter value, i.e., efficiently solvable for any fixed optimal parameter value \cite{abraham06}, so the result above is surprising in the sense that even for a fixed optimal parameter value and bounded degree graphs, the problem remains computationally intractable.

\section{Parametrised Algorithms}
\label{sec:paramalgo}

One way to cope with the strong computational intractability we established in the previous section is to explore algorithms with a parametrised time complexity. Due to Corollary \ref{cor:devmaxsmihard}, we know immediately that {\sc $k$-Deviator-Max-smi} is not in the complexity class {\sf XP} (and hence also not in {\sf FPT} \cite{flumgrohe}) with respect to the maximum preference list length $d_{\max}$, assuming that {\sf P}$\neq${\sf NP}. However, on the positive side, we will now show that the more general non-bipartite version of the problem, i.e., {\sc $k$-Deviator-Max-sri}, is fixed-parameter tractable (in the complexity class {\sf FPT}) with respect to the combined parameter $(\vert D\vert,d_{\max})$, as well as in {\sf XP} with respect to just $\vert D\vert$ (and hence the same applies also for the bipartite restriction {\sc $k$-Deviator-Max-smi}). We show this by giving an algorithm, which is efficient for small sets $D$ and small $k$ or $d_{\max}$ values (notice that $k$ is part of the problem, not an input parameter). Note that $d_{\max}$ is naturally small in most practical scenarios: as preference lists usually need to be elicited manually, there is often a limit on how many choices a participant ranks. Indeed, we are aware of multiple real-world applications of matching schemes, such as in the allocation of students to projects, students to campus housing, or pupils to courses, where participants are asked to provide a list of their 3-8 top choices only. Furthermore, notice that when maximising the matching size as a primary objective, it is often beneficial for agents to submit a short preference list, as this may increase their chance of getting a better partner / project / room. Of course, $\vert D\vert$ need not be small in general. However, one possible strategy is to fix $k=0$, in which case the algorithm terminates in $O(d_{\max}^{O(\vert D\vert)}n^{5/2})$ time, and rank agents based on their likelihood to deviate. Then, one could start by including only the strongest deviators in $D$, in which case $D$ is small, and iteratively increase the size of $D$ up until either no more solution exists in which no deviator agent is blocking, or the computation takes too long.

Intuitively, our algorithm enumerates possible partner and blocking-pair configurations for the deviator agents, and tries to extend each feasible configuration to a matching that satisfies our requirements using a maximum-weight matching subroutine. Our procedure is summarised in Algorithm \ref{alg:fpt} and operates as follows. Recall that we aim to find a maximum-cardinality matching $M_P$ of $I$ such that the agents in $D$ are collectively involved in at most $k$ blocking pairs in $M_P$. We will denote such blocking pairs by $bp(M_P)\vert_{D}$. First, in line 1, we compute a maximum-cardinality matching $M_S$ using the classical Micali-Vazirani algorithm \cite{vaziranicorrect,micalivazirani} to determine the target size of our matching.

\begin{algorithm}[!htb]

\begin{algorithmic}[1]
\Require{$I=(A,\succ)$ : an {\sc sri} instance; $D$ : deviator agents; $k$ : maximum permitted number of blocking pairs involving deviators}
\Ensure{$M$ : a solution (matching), if one exists}

\State $M_S\gets$ {\tt MV}$(I)$ \Comment{Find a max.-card. matching}
\For{each candidate matching $M_D\in \mathcal M_D$ of $D$ agents}
    \For{each candidate set $B$ of at most $k$ blocking pairs involving $D$ agents}
        \If{$M_D\cap B=\varnothing$}
            \State $I'\gets I$ \Comment{Make a copy of the instance}
            \State $Q\gets\varnothing$
            \For{each $a_i\in D$}
                \For{each $a_r\in A$ such that $a_r\succ_i M_D(a_i)$ and $\{a_i,a_r\}\notin B$}
                    \State $Q$.add$(a_r)$
                    \State remove $a_i$ and every worse agent from $\succ'_r$, if such agents exist
                \EndFor
            \EndFor
            \If{$\nexists\{a_i,a_j\}\in M_D$ such that $\{a_i,a_j\}$ unacceptable in $I'$}
                \State $V\gets A\setminus A(M_D)$ \Comment{All agents other than those matched in $M_D$}
                \State $E\gets$ acceptable pairs in $I'$ whose endpoints belong to $V$
                \State define $w(\{a_r,a_s\})=\vert A\vert+\vert\{a_r,a_s\}\cap Q\vert$ for all $\{a_r,a_s\}\in E$
                \State $M_{mw}\gets$ {\tt HK}$(V,E,w)$ \Comment{Find a max.-weight matching}
                \If{every agent in $Q$ is matched in $M_{mw}$ and $\vert M_D\cup M_{mw}\vert=\vert M_S\vert$}
                    \State\Return{$M_D\cup M_{mw}$}
                \EndIf
            \EndIf
        \EndIf
    \EndFor
\EndFor

\State\Return{no solution}

\end{algorithmic}
\caption{Exponential-time algorithm for {\sc $k$-Deviator-Max-sri}.}
\label{alg:fpt}
\end{algorithm}

Next, we consider each possible combination of choices of valid partners for agents in $D$ (including being unmatched), i.e., for each agent in $D$ we consider matching it to each of the agents on their preference list (regardless of whether the agent on the preference list is also a deviator or not). We can immediately discard any combinations of choices that are not matchings, and denote the remaining set of candidate matchings by $\mathcal{M}_D$. Next, we consider each candidate matching sequentially (in line 2ff) and aim to extend it to a solution to our problem. Let us fix a candidate matching $M_D\in \mathcal{M}_D$.

Now, we consider each possible set of blocking pairs involving at least one agent in $D$ of size at most $k$ (in line 3ff). Notice that these two nested for loops (one over the set of candidate matchings and one over the set of possible blocking pairs) have two different roles: the candidate matching fixes the partners of the deviator agents, and is tested to be part of the final solution, but each candidate matching can still admit different sets of blocking pairs depending on the remainder of the final matching. Let us fix a candidate set of blocking pairs $B$.

Clearly, no pair of agents can be simultaneously matched and blocking, so if $M_D\cap B\neq \varnothing$, then we reject this configuration $(M_D,B)$ (implicit in line 4). Otherwise, for every agent $a_i\in D$ and every agent $a_r\in A$ such that $a_r\succ_i M_D(a_i)$ and $\{a_i,a_r\}\notin B$, agent $a_r$ must end up with a partner better than $a_i$ (according to $\succ_r$). Thus, we truncate $\succ_r$ in line 10, discarding every agent worse than and including $a_i$ (if multiple agents in $D$ satisfy this criterion with respect to $a_r$, we truncate at the best-ranked such agent according to $\succ_r$). We must keep track of these agents $a_r$ to later verify that they do indeed have a partner, and denote the set of all such agents $a_r$ by $Q$, which is constructed in line 9.

If, after carrying out these truncations, some agent $a_i\in D$ has a partner in $M_D$ that is no longer in their preference list (or vice versa), we can reject this configuration $(M_D,B)$ (implicit in line 13). Otherwise, we need to look for a matching of maximum size among the agents not yet matched in $M_D$ such that every agent in $Q$ is matched. We can do this as follows: we construct a maximum-weight matching instance $(G,w)$ (where $G$ is a graph and $w$ is a weight function from the set of edges to the positive integers) consisting of all agents $A\setminus A(M_D)$ (where $A(M_D)$ denotes the set of agents matched in $M_D$) as vertices and all acceptable pairs among these agents as edges. Furthermore, we construct $w$ as follows: for every acceptable pair of agents $a_r,a_s$ present in the instance, we let $w(\{a_r,a_s\})=n + \vert\{a_r,a_s\}\cap Q\vert$, i.e., we assign a weight of $n$ to every edge, plus one extra unit for each endpoint in $Q$. This is done in lines 14-16.

Now, in line 17, we find a maximum-weight matching $M_{mw}$ with weight $w(M_{mw})$ for the instance $(G,w)$ using the algorithm by \citet{huangkavitha}. We claim (and prove this in the theorem below) that $M_{mw}$ is a maximum-cardinality matching of $(G,w)$ in which every agent in $Q$ is matched if and only if such a matching exists. Thus, if there exists an agent $a_r\in Q$ that is not matched in $M_{mw}$, we can reject this configuration $(M_D,B)$. Otherwise, by construction, $M_D$ and $M_{mw}$ are disjoint, so if $\vert M_D\cup M_{mw}\vert = \vert M_S\vert$, then we can accept $M=M_D\cup M_{mw}$ as our solution (line 19). Otherwise, $M_D$ does not extend to a solution with respect to $B$, and we can reject this configuration $(M_D,B)$. This completes the description of the algorithm. We now state and prove the correctness and the complexity of our algorithm formally.

\begin{theorem}
\label{thm:preferentialbipfpt}
    Let $(I=(A,\succ),D)$ be an instance of {\sc $k$-Deviator-Max-sri} with $n$ agents and maximum preference list length $d_{\max}$. Then the problem {\sc $k$-Deviator-Max-sri} is solvable in $O(\vert D\vert^{k}d_{\max}^{\vert D\vert+k+1}n^{5/2})$ time.
\end{theorem}
\begin{proof}
    We start by proving our claim that $M_{mw}$, as computed in line 17 of Algorithm \ref{alg:fpt}, is a maximum-cardinality matching of $(G,w)$ in which every agent in $Q$ is matched if and only if such a matching exists. First, note that clearly $M_{mw}$ is a maximum-cardinality matching of $G$, i.e., it matches the maximum possible number $p$ of agents altogether, $q$ of which are agents in $Q$. Suppose not, i.e., suppose that there exists a matching $M_{mw}'$ that matches strictly more (total) agents $p'>p$, $q'$ of which are agents in $Q$. Clearly $w(M_{mw}')= p'n+q'$ and $w(M_{mw})=pn+q$. Because $M_{mw}$ is a maximum-weight matching, it must be the case that $w(M_{mw}')\leq w(M_{mw})$. Then $p'n+q'\leq pn+q$, and after rearranging, $(p'-p)n\leq q-q'$. However, by assumption that $p'>p$, then $n\leq (p'-p)n$, so $n\leq q-q'$. However, because $q\leq n$ and $q'\geq 0$ by construction of $G$, then it must be the case that $q=n$, i.e., $Q=A$, and therefore $M_{mw}$ matches all $n$ agents in $G$, which contradicts that $p'>p$. Thus, $M_{mw}$ is indeed a maximum-cardinality matching. Now, suppose for the sake of contradiction that there exists another maximum-cardinality matching $M_{mw}'$ that also matches a total of $p$ agents, but matches strictly more agents $q'>q$ that are in $Q$. Then $w(M_{mw}')= pn+q'>pn+q=w(M_{mw})$, contradicting that $M_{mw}$ is a maximum-weight matching.

    Further, to prove correctness, we note that if we accept $M$, then $M$ is a maximum-cardinality matching of $I$. By construction, $\vert bp(M)\vert_{D}\vert \leq \vert B\vert\leq k$. Thus, $M$ is a solution to {\sc $k$-Deviator-Max-sri}. Otherwise, if, for every configuration $(M_D,B)$ and an associated maximum-weight matching $M_{mw}$, it is the case that either $\vert M_D\cup M_{mw}\vert<\vert M_S\vert$ or not every agent in $Q$ is matched, then it must be true that $(I,D)$ is a no-instance to {\sc $k$-Deviator-Max-smi}, i.e., that there exists no solution to the problem for this instance. To see this, notice that if $(I=(A,\succ), D)$ has a solution $M_P$, then $M_P$ can be divided into a part $M_1$ in which every pair contains at least one agent from $D$, and a part $M_2$ which does not contain any pair with an agent from $D$. Clearly, $M_1$ satisfies the requirements of the candidate matchings $\mathcal{M}_D$, so we would have considered it in line 2, and $M_2$ must be a maximum-cardinality matching among the remaining agents (otherwise $M_P$ would not be a maximum-cardinality matching). Let us fix this $M_1\in \mathcal{M}_D$. By being a solution, it must furthermore be true that, for every agent $a_i\in D$ and every agent $a_r\in A$ such that $a_r\succ_i M_P(a_i)$, either $\{a_i,a_r\}\in bp(M_P)\vert_{D}$ or ($\{a_i,a_r\}\notin bp(M_P)\vert_{D}$ and $M_P(a_r)\succ_r a_i$). We keep track of such agents $a_r$ that satisfy the second case in the set $Q$. It must also be the case that $\vert bp(M_P)\vert_{D}\vert=\vert bp(M_1\cup M_2)\vert_{D}\vert \leq k$, so we would have considered this set $bp(M_P)\vert_{D}$ as a candidate set $B$ in line 3. Furthermore, $M_2$ must be a maximum-cardinality matching of our maximum-weight matching instance $(G,w)$ (for some configuration $(M_D,B)$) that matches all such agents $Q$, and, by the key property we stated about maximum-weight matchings, $M_2$ must be a maximum-weight matching of $(G,w)$, too. Thus, our algorithm would have identified this solution, or a different solution, to the problem.
    
    Finally, we expand on the time complexity analysis. Clearly, every agent in $a_i\in D$ is either matched to one of the at most $d_{\max}$ many agents on their preference list or remains unmatched. Thus, there are $(d_{\max}+1)^{\vert D\vert}=O(d_{\max}^{\vert D\vert})$ possible combinations of choices of partner, including the possibility of being unmatched, for the agents in $D$. Also, there are clearly at most $\sum_{1\leq r\leq k}\binom{\vert D\vert d_{\max}}{r}$ many possible sets of blocking pairs involving some agent in $D$. We can apply the loose upper bound $\binom{\vert D\vert d_{\max}}{r}\leq (\vert D\vert d_{\max})^r$ and conclude that $\sum_{1\leq r\leq k} (\vert D\vert d_{\max})^r=\frac{(\vert D\vert d_{\max})^{k+1}-1}{\vert D\vert d_{\max}-1}$ by the geometric series. Furthermore, the bound $\frac{(\vert D\vert d_{\max})^{k+1}-1}{\vert D\vert d_{\max}-1}\leq \frac{(\vert D\vert d_{\max})^{k+1}}{\vert D\vert d_{\max}-1}\leq 2(\vert D\vert d_{\max})^{k}=O((\vert D\vert d_{\max})^{k})$ applies. 

    We can construct $(G,w)$ in linear time in the size of $I$, so in $O(d_{\max}n)$ time. The algorithm by \citet{huangkavitha} runs in $O(\sqrt{n_*}m_*N\log(n_*^2/m_*)/\log(n_*))$ time, where $n_*$ is the number of vertices in the graph, $m_*$ is the number of edges in the graph and $N$ is the maximum edge weight. Thus, given that $n_*\leq n$ $m_*\leq d_{\max}n$ and $N\leq n+2$, this algorithm runs in $O(d_{\max}n^{5/2})$ time.    

    Putting the complexity analysis together, we consider $O(d_{\max}^{\vert D\vert})$ candidate matchings and, for every candidate matching, we consider $O((\vert D\vert d_{\max})^{k})$ sets of blocking pairs, leading to at most $O(d_{\max}^{\vert D\vert}(\vert D\vert d_{\max})^{k})=O(\vert D\vert^{k}d_{\max}^{\vert D\vert+k})$ candidate configurations. For each candidate configuration, we construct a maximum-weight matching instance in $O(d_{\max}n)$ time and, subsequently, compute a maximum-weight matching in $O(d_{\max}n^{5/2})$ time. Hence, we arrive at a final complexity of $O(\vert D\vert^{k}d_{\max}^{\vert D\vert+k+1}n^{5/2})$.
\end{proof}

We remark that the bipartite restriction {\sc $k$-Deviator-Max-smi} can be solved slightly faster than stated above (by at least a factor of $\frac{n}{\log n}$) by replacing the maximum-weight matching algorithm by \citet{huangkavitha} with the algorithm by \citet{duansu12}. 

Using our result in Theorem \ref{thm:preferentialbipfpt}, we can conclude the following complexity classifications for our decision and optimisation problems of interest.

\begin{corollary}
\label{cor:preffptxp}
    {\sc $k$-Deviator-Max-sri} is in {\sf FPT} with respect to $(\vert D\vert,d_{\max})$ and in {\sf XP} with respect to $\vert D\vert$. Furthermore, {\sc Deviator-Max-sri} is in {\sf FPT} with respect to $(\vert D\vert,d_{\max})$ and in {\sf XP} with respect to $(\vert D\vert,opt)$, where $opt=\min_{M\in \mathcal{M}^+}\vert\bigcup_{a_i\in D}bp_{a_i}(M)\vert$.
\end{corollary}
\begin{proof}
    From Theorem \ref{thm:preferentialbipfpt}, we know that the problem is solvable in $O(\vert D\vert^{k}d_{\max}^{\vert D\vert+k+1}n^{5/2})$ time. Given that $k$ is part of the problem definition, not part of the input, we consider $k=O(1)$. Alternatively, we can justify this as follows: the optimal parameter value $\vert \sum_{a_i\in D}bp_{a_i}(M)\vert$ is bounded from above by $\vert D\vert d_{\max}$. Therefore, for any $k\geq\vert D\vert d_{\max}$, the problem becomes trivial. Thus, we may assume without loss of generality that $k<\vert D\vert d_{\max}$, in which case $k$ does not need to be treated as a parameter when $\vert D\vert, d_{\max}$ are fixed. Hence, the first {\sf FPT} result follows directly from $O(\vert D\vert^{O(1)}d_{\max}^{\vert D\vert+O(1)}n^{5/2})$. For the first {\sf XP} result, notice that $d_{\max}\leq n-1< n$, so we can crudely upper-bound the complexity by $O(\vert D\vert^{O(1)}n^{O(\vert D\vert)})$ as required. Finally, to solve {\sc Deviator-Max-sri}, we can simply solve {\sc $k$-Deviator-Max-sri} with increasing values of $k$ up until a solution is found, i.e., for instance $(I,D)$, we iterate $0\leq k\leq \min_{M\in \mathcal{M}^+}\vert\bigcup_{a_i\in D}bp_{a_i}(M)\vert\leq \vert D\vert d_{\max}$ and find a solution in $O(\vert D\vert d_{\max}\vert D\vert^{\vert D\vert d_{\max}}d_{\max}^{\vert D\vert+\vert D\vert d_{\max}+1}n^{5/2})=O(\vert D\vert^{O(\vert D\vert d_{\max})}d_{\max}^{O(\vert D\vert d_{\max})}n^{5/2})$ time. Furthermore, $O(k\vert D\vert^{k}d_{\max}^{\vert D\vert+k+1}n^{5/2})=O(k\vert D\vert^{O(k)}n^{O(\vert D\vert+k)})$.
\end{proof}

We will now show how to speed up the algorithm when dropping the maximum-cardinality requirement. Notice that the fact that we do not require a maximum-cardinality matching allows us to arrive at a complexity independent of the total number of agents $n$. This is because any agent ``far enough'' from agents in $D$ may simply be left unmatched. In a sense, this gives a sort-of kernelisation of the problem, where we only need to consider a connected neighbourhood around agents in $D$ to solve the problem. We performed the analysis assuming that $d_{\max}\ll n$, which leads to a larger exponent for $d_{\max}$, but it is easy to refine the analysis to a smaller exponent in exchange for a sub-quadratic dependency on $n$ when this is not the case.

\begin{restatable}{theorem}{preferentialfpt}
\label{thm:preferentialfpt}
    Let $(I=(A,\succ),D)$ be an instance of {\sc $k$-Deviator-sri} with maximum preference list length $d_{\max}$. Then the problem {\sc $k$-Deviator-sri} is solvable in $O(\vert D\vert^{k+3/2}d_{\max}^{\vert D\vert+k+4})$ time.
\end{restatable}
\begin{proof}[Proof sketch]
    The high-level approach is the same as in Algorithm \ref{alg:fpt} and the proof of Theorem \ref{thm:preferentialbipfpt}. We highlight critical differences below and refer to Appendix \ref{apx:kdevsri}.

    As we no longer require a maximum-cardinality matching, we can ignore line 1 of Algorithm \ref{alg:fpt}. In lines 14-17 of Algorithm \ref{alg:fpt}, when searching for a maximum-cardinality matching among the agents not yet matched in $M_D$ such that every agent in $Q$ is matched, we can drop the maximum-cardinality requirement. This simplifies the construction of $(G,w)$ because we can ignore what happens to agents that are ``far enough'' away from agents in $D$, as they do not influence the stability of agents in $D$. Agents within neighbourhood distance 2 can influence the stability of agents in $D$, though: if an agent $r\in D$ is matched to some agent $x\notin D$, but $r$ prefers $y\notin D$ to $x$, then either $y$ is blocking with $x$, or $y$ gets a partner $z$ better than $x$. Notice that $z$ is in $N_2[D]$ but might not be in $N_1[D]$. However, we can simply leave all agents outside of this radius unmatched. Hence, we let the graph $G$ consist of all agents that are at most two agents away from some agent in $D$ (i.e., within neighbourhood distance 2 of some agent in $D$ in the acceptability graph of $I$) and not yet matched in $M_D$ as vertices, and all acceptable pairs among these agents as edges. This affects line 14 in Algorithm \ref{alg:fpt} (which should be changed to $V\gets N_2[D]\setminus A(M_D)$). Furthermore, we change line 16 of Algorithm \ref{alg:fpt} such that $w$ is defined as follows: for every acceptable pair of agents $a_r,a_s$ present in the instance, we let $w(\{a_r,a_s\})=\vert\{a_r,a_s\}\cap Q\vert$, i.e., we assign a weight of 1 for every endpoint in $Q$. 
    
    We can construct $(G,w)$ in linear time in the graph we create, assuming that preference lists of any given agent can be explored efficiently. By the preference list lengths, there are at most $\vert D\vert d_{\max}^2$ many agents that are at most two agents away from some agent in $D$ (and, as argued above, contrary to the algorithm in Theorem \ref{thm:preferentialbipfpt}, we can ignore all agents outwith the neighbourhood distance two of any agent in $D$), and $G$ has at most $\vert D\vert d_{\max}^3$ edges. Thus, we can create $(G,w)$ in $O(\vert D\vert d_{\max}^3)$ time. Now we find a maximum-weight matching $M_{mw}$ using the algorithm by Huang and Kavitha \cite{huangkavitha}, which runs in $O(\sqrt{n_*}m_*N\log(n_*^2/m_*)/\log(n_*))=O(\vert D\vert^{3/2}d_{\max}^4)$ time. Leftover conformists that are distance at least 3 away from any deviator can simply be left unmatched.

    If there exists an agent $a_r\in Q$ that is not matched in $M_{mw}$, we can reject this configuration $(M_D,B)$. Otherwise, by construction, $M_D$ and $M_{mw}$ are disjoint, so we can accept $M=M_D\cup M_{mw}$ as our solution. This simplifies the condition in line 18, as we no longer need to verify the size of the matching.
\end{proof}

Notice that dropping the maximum-cardinality requirement here results in a time complexity of at most $O(\vert D\vert^{k+3/2}d_{\max}^{\vert D\vert+k+4})$ compared to the time complexity of at most $O(\vert D\vert^{k}d_{\max}^{\vert D\vert+k+1}n^{5/2})$ we concluded in Theorem \ref{thm:preferentialbipfpt} for {\sc $k$-Deviator-Max-sri}, which is more efficient whenever $n$ is much larger than $D$ and $d_{\max}$.

Again, due to a similar argument as in Corollary \ref{cor:preffptxp}, we can immediately note the following complexity classification.

\begin{corollary}
\label{cor:psrifpt}
    {\sc $k$-Deviator-sri} is in {\sf FPT} with respect to $(\vert D\vert,d_{\max})$ and in {\sf XP} with respect to $\vert D\vert$. Furthermore, {\sc Deviator-sri} is in {\sf FPT} with respect to $(\vert D\vert, d_{\max})$ and in {\sf XP} with respect to $(\vert D\vert, opt)$, where $opt=\min_{M\in \mathcal M}\vert\bigcup_{a_i\in D} bp_{a_i}(M)\vert$.
\end{corollary}

\section{Minimising the Number of Blocking Deviators Instead}
\label{sec:prefminba}

As discussed in Section \ref{sec:related}, an alternative approach to minimising blocking pairs is to minimise the number of blocking agents. There is an inherent trade-off: minimising blocking pairs minimises the aggregate instability of the matching, but might involve a higher total number of blocking agents that have an individual incentive to deviate. Minimising blocking agents instead aims to give as few agents as possible an incentive to deviate, although their individual incentives or opportunities to deviate might be higher due to possibly more blocking pairs that they are a part of. As a third option, minimax almost-stable matchings, which were recently investigated by \citet{glitznermanloveminmax}, minimise the maximum individual incentive to deviate.

In this section, we will briefly highlight that minimising the number of blocking deviators remains {\sf NP-hard}. However, we will also show how to extend the {\sf FPT} algorithms to this setting. Let {\sc Deviator-BA-sri} and {\sc Deviator-BA-Max-sri} denote versions of {\sc Deviator-sri} and {\sc Deviator-Max-sri} that minimise $\vert\{a_i\in D\;\vert\;bp_{a_i}(M)\neq \varnothing\}\vert$ among all matchings and all maximum-cardinality matchings, respectively. Furthermore, let {\sc $k$-Deviator-BA-sri} and {\sc $k$-Deviator-BA-Max-sri} denote the problems of computing a matching and a maximum-cardinality matching admitting at most $k$ blocking agents in $D$, if they exist, respectively. Also, let {\sc 0-Deviator-BA-sri-Dec} and {\sc 0-Deviator-BA-Perfect-smi-Dec} denote the restricted problems of deciding whether such a solution exists when $k=0$ and, in the {\sc smi} case, where maximum-cardinality matchings are perfect, and the instance is bipartite. We can immediately conclude the following.

\begin{theorem}
\label{thm:prefBAhard}
    {\sc 0-Deviator-BA-sri-Dec} is {\sf NP-complete}, regardless of whether preference lists are of length at most 5 or complete. {\sc 0-Deviator-BA-Perfect-smi-Dec} is {\sf NP-complete}, even if all preference lists are of length at most 3 and all deviators are on the same side.
\end{theorem}
\begin{proof}
    Clearly, both problems are in {\sf NP}: we can easily check a possible solution $M$ by iterating through the set of deviator agents and checking that none of them block with any agent that they prefer to their partner in $M$ (which might be none). For the {\sc Perfect-smi} case, we furthermore need to check that every agent in the instance has a partner in $M$, which can be done efficiently by iterating through the set of agents once.
    
    Now, let $I=(A,\succ)$ be an {\sc sri} instance and let $D\subseteq A$. Then there exists a matching $M\in\mathcal M$ such that $\vert\{a_i\in D\;\vert\;bp_{a_i}(M)\neq \varnothing\}\vert=0$ if and only if $\bigcup_{a_i\in D}bp_{a_i}(M)=\varnothing$. Hence, $(I,D)$ is a yes-instance to {\sc 0-Deviator-BA-sri-Dec} if and only if it is a yes-instance to {\sc 0-Deviator-sri-Dec}, which we showed to be {\sf NP-complete} regardless of whether preference lists are of length at most 5 or complete in Theorem \ref{thm:prefsrihard}.

    Similarly, let $I=(A,\succ)$ be an {\sc smi} instance and let $D\subseteq A$. Then there exists a matching $M\in\mathcal{M}^+$ such that $\vert\{a_i\in D\;\vert\;bp_{a_i}(M)\neq \varnothing\}\vert=0$ if and only if $\bigcup_{a_i\in D}bp_{a_i}(M)=\varnothing$. Hence, $(I,D)$ is a yes-instance to {\sc 0-Deviator-BA-Perfect-smi-Dec} if and only if it is a yes-instance to {\sc 0-Deviator-Perfect-smi-Dec}, which we showed to be {\sf NP-complete} even in the case where preference lists are of length at most 3 and all deviators are on the same side in Theorem \ref{thm:preferentialbip}.
\end{proof}

Thus, we can immediately conclude the following intractability results for our desired optimisation problems.

\begin{corollary}
\label{cor:devbahard}
    {\sc Deviator-BA-sri} and {\sc Deviator-BA-Max-smi} are {\sf NP-hard} and are not in {\sf XP} (unless {\sf P}={\sf NP}) with respect to optimum values $\kappa=\min_{M\in\mathcal M}\vert\{a_i\in D\;\vert\;bp_{a_i}(M)\neq \varnothing\}\vert$ and $\kappa=\min_{M\in\mathcal M^+}\vert\{a_i\in D\;\vert\;bp_{a_i}(M)\neq \varnothing\}\vert$, respectively.
\end{corollary}
    
On the positive side, the {\sf FPT} algorithms that minimise blocking pairs can be adapted naturally to minimise blocking agents instead, as we show below. We start with an algorithm for {\sc sri} instances and no requirements on the matching size. Notice that the bound on the runtime is faster than the one in Theorem \ref{thm:preferentialfpt} by a factor of $d_{\max}^k$, and we point out the reason for this in the following proof.

\begin{restatable}{theorem}{preferentialBAfpt}
\label{thm:preferentialBAfpt}
    Let $(I=(A,\succ),D)$ be an instance of {\sc $k$-Deviator-BA-sri} with maximum preference list length $d_{\max}$. Then the problem {\sc $k$-Deviator-BA-sri} is solvable in $O(\vert D\vert^{k+3/2}d_{\max}^{\vert D\vert+4})$ time.
\end{restatable}
\begin{proof}[Proof sketch]
    The high-level approach remains the same as in the proof of Theorem \ref{thm:preferentialfpt}. We highlight critical differences below and refer to Appendix \ref{apx:kdevbasri} for full details. We consider $(d_{\max}+1)^{\vert D\vert}=O(d_{\max}^{\vert D\vert})$ possible combinations of choices of partner, including the possibility of being unmatched, for the agents in $D$. Then, instead of enumerating the candidate sets of blocking pairs as in line 3 of Algorithm \ref{alg:fpt}, we consider each possible set of blocking \textbf{agents} in $D$ of size at most $k$ instead. 
    There are at most $\sum_{1\leq r\leq k}\binom{\vert D\vert}{r}=O(\vert D\vert^{k})$ many such sets. From here, the approach is the same again as in the proof of Theorem \ref{thm:preferentialfpt}, making only minor adjustments to lines 4 (which should be removed) and 8 (for each $a_r\in A$ such that $a_r\succ_i M_D(a_i)$ and $a_i\notin B$) of Algorithm \ref{alg:fpt}.
\end{proof}

We now show how to adapt the {\sf FPT} algorithm for the {\sc Max-sri} case. Also, for this algorithm, the bound on the runtime is faster than the one in Theorem \ref{thm:preferentialbipfpt} by a factor of $d_{\max}^k$, for the same reason as in the proof of Theorem \ref{thm:preferentialBAfpt}.

\begin{restatable}{theorem}{preferentialBAbipfpt}
\label{thm:preferentialBAbipfpt}
    Let $(I=(A,\succ),D)$ be an instance of {\sc $k$-Deviator-BA-Max-sri} with $n$ agents and maximum preference list length $d_{\max}$. Then the problem {\sc $k$-Deviator-BA-Max-sri} is solvable in $O(\vert D\vert^{k}d_{\max}^{\vert D\vert+1}n^{5/2})$ time.
\end{restatable}
\begin{proof}[Proof sketch]
    Again, the high-level approach remains the same as in the proofs of Theorems \ref{thm:preferentialbipfpt} and \ref{thm:preferentialBAfpt}. We highlight critical differences below and refer to Appendix \ref{apx:kdevbamaxsri} for full details. We consider at most $(d_{\max}+1)^{\vert D\vert}=O(d_{\max}^{\vert D\vert})$ possible candidate matchings $\mathcal{M}_D$ and at most $O(\vert D\vert^k)$ many possible sets of blocking agents in $D$. Similarly, for each combination of candidate matching $M_D$ and candidate set of blocking agents $B$, we create the set $Q$ of agents that must be matched in a corresponding solution as in the proof of Theorem \ref{thm:preferentialBAfpt}. This is achieved using the weight function already present in line 16 of Algorithm \ref{alg:fpt}.
\end{proof}

With these results, we may conclude the following complexity classification for our optimisation problems of interest.

\begin{corollary}
\label{cor:prefba}    
    The problems {\sc $k$-Deviator-BA-sri} and {\sc $k$-Deviator-BA-Max-sri} are in {\sf FPT} with respect to $(\vert D\vert,d_{\max})$ and in {\sf XP} with respect to $\vert D\vert$. Furthermore, {\sc Deviator-BA-sri} and {\sc Deviator-BA-Max-sri} are in {\sf FPT} with respect to $(\vert D\vert, d_{\max})$ and in {\sf XP} with respect to $(\vert D\vert, opt)$, where $opt=\min_{M\in \mathcal M}\vert ba(M)\cap D\vert$ and $opt=\min_{M\in \mathcal {M}^+}\vert ba(M)\cap D\vert$, respectively.
\end{corollary}

\section{Efficient Algorithms for Instances with Short Lists}
\label{sec:prefshort}

In addition to the {\sf FPT} algorithms for our collection of deviator-stable matching problems, we will now show that all these problems are solvable in polynomial time whenever preference lists are of length at most 2. Although this setting is, of course, highly restricted, these algorithms can still present practical tools to compute optimal solutions whenever choices are very limited. Furthermore, the algorithms would also be applicable when agents are asked to only provide their two favourite choices, after which we first find a deviator-stable maximum-cardinality matching, and then arbitrarily assign the remaining agents to each other. While not giving an optimal solution, this approach does ensure desirable stability guarantees with respect to the top 2 choices of agents, which agents likely feel the most strongly about anyway, and, subject to this, minimises the number of unmatched agents. From a computational complexity perspective, the new algorithms also establish tight and almost-tight complexity dichotomies for our problems of interest.

We start with Algorithm \ref{alg:prefexact2sri}, which returns an optimal solution to the problems {\sc Deviator-sri} and {\sc Deviator-BA-sri} in polynomial time. Notice that because preference lists are of length at most 2, the acceptability graph of $I$ has maximum degree at most 2, i.e., it consists only of paths and cycles. Algorithm \ref{alg:prefexact2sri} exploits this as follows: every path component involving any number of agents and every cycle component involving an even number of agents is bipartite, so we can run the classical Gale-Shapley algorithm \cite{gale_shapley} on these components (we use the notation $A(C)$ to denote the set of agents involved in a component $C$ and the notation $I\vert_{A(C)}$ to denote the part of instance $I$ involving the agents $A(C)$) to find stable matchings among the involved agents. The trickier case occurs for cycles involving an odd number of agents, as these are, of course, not bipartite and might not admit stable matchings. Here, we include a stable matching within the component, if one exists, and, otherwise, include a maximum-cardinality matching that leaves exactly one agent unmatched -- which agent to exclude depends on the occurrences of deviator agents in the cycle.

\begin{algorithm}[!htb]

\begin{algorithmic}[1]
\Require{$I$ : an {\sc sri} instance}
\Require{$D$ : deviator agents}
\Ensure{$M$ : a matching}

\State $M \gets \varnothing$ \Comment{Initialise an empty solution}

\For{each path $P=(a_{r_1}\;a_{r_2}\;a_{r_3}\dots a_{r_k})$ of $I$ involving $k$ agents}
    \State $M\gets M\;\cup$ {\sf GaleShapley}$(I\vert_{A(P)})$
\EndFor
\For{each cycle $C=(a_{r_1}\;a_{r_2}\;a_{r_3}\dots a_{r_k})$ of $I$ involving an even number of $k$ agents}
    \State $M\gets M\;\cup$ {\sf GaleShapley}$(I\vert_{A(C)})$
\EndFor
\For{each cycle $C=(a_{r_1}\;a_{r_2}\;a_{r_3}\dots a_{r_k})$ of $I$ involving an odd number of $k$ agents}
    \State $M_t\gets$ {\sf Irving}$(I\vert_{A(C)})$ \Comment{Find a stable matching, if one exists}
    \If{$M_t$ exists}
        \State $M\gets M\cup M_t$
    \Else
        \If{there exist $a_{r_u},a_{r_v}$ in $A(C)\setminus D$ such that $a_{r_v}$ is $a_{r_u}$'s first choice}
            \State $M\gets M\cup \{\{a_{r_{v+2x-1}},a_{r_{v+2x}}\}\;\vert\;1\leq x\leq \lfloor\frac{k}{2}\rfloor\}$ \Comment{Leave $a_{r_v}$ unmatched}
        \ElsIf{there exist $a_{r_u}\in A(C)\setminus D$, $a_{r_v} \in A(C)\cap D$ such that $a_{r_v}$ is $a_{r_u}$'s first choice}
            \State $M\gets M\cup \{\{a_{r_{v+2x-1}},a_{r_{v+2x}}\}\;\vert\;1\leq x\leq \lfloor\frac{k}{2}\rfloor\}$ \Comment{Leave $a_{r_v}$ unmatched}
        \Else
            \State $M\gets M\cup \{\{a_{r_1},a_{r_2}\},\{a_{r_3},a_{r_4}\},\dots,\{a_{r_{k-2}},a_{r_{k-1}}\}\}$ \Comment{Leave $a_{r_k}$ unmatched}
        \EndIf
    \EndIf
\EndFor

\State\Return{$M$}

\end{algorithmic}
\caption{Exact algorithm for {\sc sri} problems with lists of length at most 2}
\label{alg:prefexact2sri}
\end{algorithm}

\begin{theorem}
\label{thm:prefexact2sri}
    Let $I=(A,\succ)$ be an {\sc sri} instance with $n$ agents and let $D\subseteq A$ be a set of deviators. If all preference lists of $I$ are of length at most 2, then Algorithm \ref{alg:prefexact2sri} computes an optimal solution to {\sc Deviator-sri} and {\sc Deviator-BA-sri} in $O(n)$ time.
\end{theorem}
\begin{proof}
    Let $M$ be the matching returned by Algorithm \ref{alg:prefexact2sri}. First, notice that distinct connected components have -- by definition -- independent preference systems, so no two agents from distinct connected components can form a blocking pair together. Hence, any blocking pairs and any blocking agents admitted by $M$ must stem from our treatment of cycles $C$ involving an odd number of agents, because all other components are bipartite and must admit stable matchings, which we include in $M$.

    Now, let $C=(a_{r_1}\;a_{r_2}\;a_{r_3}\dots a_{r_k})$ be a cycle involving an odd number of agents $k$. If $C$ admits a stable matching, then we find it using Irving's algorithm \cite{irving_sr}. Otherwise, $C$ must be an ordered cycle such that either $a_{r_{i+1}}\succ_{r_i}a_{r_{i-1}}$ or $a_{r_{i-1}}\succ_{r_i}a_{r_{i+1}}$ for all $1\leq i\leq k$ (addition and subtraction taken modulo $k$) \cite{tan91_1}. Without loss of generality, suppose that the former case holds, i.e., that every agent prefers their successor over their predecessor (as indicated by the increasing index within $C$). Then, when we include one of the maximum-cardinality matchings of $C$, i.e., a matching that leaves exactly one $a_{r_j}$ agent unmatched, the minimum number of (two) blocking agents (one blocking pair) is created. Hence, to minimise the number of blocking $D$ agents contained in $C$ (or to minimise the number of blocking pairs involving agents in $D$), we shall include a maximum-cardinality matching that leaves an agent $a_{r_j}\notin D$ unmatched ($a_{r_v}$ in the algorithm), such that also $a_{r_{j-1}}\notin D$ ($a_{r_u}$ in the algorithm), if such a pair of agents exists within $C$, causing no blocking $D$ agents and no blocking pairs involving agents in $D$. Otherwise, if no such pair of agents exists in $C$, we leave any one agent $a_{r_j}\notin D$ unmatched ($a_{r_v}$ in the algorithm), if one exists within $C$, causing one blocking $D$ agent and one blocking pair involving agents in $D$. Otherwise, if all agents within $C$ are also in $D$, then we can leave an arbitrary agent from $C$ unmatched ($a_{r_k}$ in the algorithm), causing two blocking $D$ agents and one blocking pair involving agents in $D$. By construction, this is optimal.
    
    Asymptotically, the Gale-Shapley algorithm and Irving's algorithm require at most linear time in the sum of the preference list lengths of the components \cite{gale_shapley,irving_sr}. As preference lists are of bounded length, this is thus linear in the number of agents of the components. Our treatment of odd-length cycles also requires linear time in the number of agents within the cycle. Each component is dealt with exactly once, and before running the algorithm, we may create three different stacks for the three different types of components by exploring the instance in linear time of $I$. Hence, the algorithm runs in at most linear time in the number of agents, i.e., in $O(n)$ time, as required.
\end{proof}

This positive result implies an almost-tight tractability frontier: the problems {\sc Deviator-sri} and {\sc Deviator-BA-sri} can both be solved very efficiently when preference lists are of length at most 2. However, the problems are {\sf NP-hard} when preference lists are of length at most 5, as previously established in Corollaries \ref{cor:devsrihard} and \ref{cor:devbahard}, respectively.

Now, for the maximum-cardinality variants of our problems, we provide Algorithm \ref{alg:prefexact2smi} for {\sc sri} instances of length at most 2, which returns an optimal solution to both {\sc Deviator-Max-sri} and {\sc Deviator-BA-Max-sri} in polynomial time. Again, we only need to deal with paths and cycles of agents. Furthermore, any path involving an even number of agents admits only one maximum-cardinality matching, so this needs to be included in any solution to these problems. Notice that any cycle involving an odd number of $k$ agents admits $k$ different maximum-cardinality matchings, each of which is uniquely determined by the agent that remains unmatched. We include the one that minimises our target. All cycles involving an even number of agents admit only two different maximum-cardinality matchings, so we include the one in the final matching that minimises our target. Lastly, paths that involve an odd number of agents $k$ require some care: each such path admits $\frac{k+1}{2}$ different maximum-cardinality matchings, so we consider each of them and include the one in the final matching that minimises our target. Below, we argue the correctness and time complexity of this approach.

\begin{algorithm}[!htb]

\begin{algorithmic}[1]
\Require{$I$ : an {\sc sri} instance}
\Require{$D$ : deviator agents}
\Ensure{$M$ : a matching}

\State $M \gets \varnothing$ \Comment{Initialise an empty solution}

\For{each path $P=(a_{r_1}\;a_{r_2}\;a_{r_3}\dots a_{r_k})$ of $I$ involving an even number of $k$ agents}
    \State $M\gets M\cup \{\{a_{r_1},a_{r_2}\},\{a_{r_3},a_{r_4}\},\dots,\{a_{r_{k-1}},a_{r_k}\}\}$ \Comment{Add the only max-card matching}
\EndFor
\For{each path $P=(a_{r_1}\;a_{r_2}\;a_{r_3}\dots a_{r_k})$ of $I$ involving an odd number of $k$ agents}
    \State $M_t\gets \{\{a_{r_{1+2x-1}},a_{r_{1+2x}}\}\;\vert\;1\leq x\leq \lfloor\frac{k}{2}\rfloor\}$ \Comment{Leave $a_{r_1}$ unmatched}
    \For{each $3\leq s\leq k$ in steps of 2}
        \State $M_t'\gets \{\{a_{r_{s+2x-1}},a_{r_{s+2x}}\}\;\vert\;1\leq x\leq \lfloor\frac{k}{2}\rfloor\}$ \Comment{Leave $a_{r_s}$ unmatched}
        \If{$\vert \bigcup_{a_i\in D}bp_i(M_t') \vert < \vert \bigcup_{a_i\in D}bp_i(M_t) \vert$} \Comment{Replace with $\vert ba(M_t')\cap D\vert < \vert ba(M_t)\cap D\vert$ to minimise the number of blocking $D$ agents instead}
            \State $M_t\gets M_t'$
        \EndIf
    \EndFor
    \State $M\gets M\cup M_t$
\EndFor
\For{each cycle $C=(a_{r_1}\;a_{r_2}\;a_{r_3}\dots a_{r_k})$ of $I$ involving an even number of $k$ agents}
    \State $M_1\gets \{\{a_{r_1},a_{r_2}\},\{a_{r_3},a_{r_4}\},\dots,\{a_{r_{k-1}},a_{r_k}\}\}$
    \State $M_2\gets \{\{a_{r_1},a_{r_k}\},\{a_{r_2},a_{r_3}\},\dots,\{a_{r_{k-2}},a_{r_{k-1}}\}\}$
    \If{$\vert \bigcup_{a_i\in D}bp_i(M_1) \vert \leq \vert \bigcup_{a_i\in D}bp_i(M_2) \vert$} \Comment{Replace with $\vert ba(M_1)\cap D\vert \leq \vert ba(M_2)\cap D\vert$ to minimise the number of blocking $D$ agents instead}
        \State $M\gets M \cup M_1$
    \Else
        \State $M\gets M \cup M_2$
    \EndIf
\EndFor
\For{each cycle $C=(a_{r_1}\;a_{r_2}\;a_{r_3}\dots a_{r_k})$ of $I$ involving an odd number of $k$ agents}
    \State $M_t\gets \{\{a_{r_2},a_{r_3}\},\{a_{r_4},a_{r_5}\},\dots,\{a_{r_{k-1}},a_{r_k}\}\}$
    \For{each $2\leq s\leq k$}
        \State $M_t'\gets \{\{a_{r_{s+x+1}},a_{r_{s+x+2}}\}\;\vert\; 0\leq x\leq k-2 \text{ in steps of 2}\}$
        \If{$\vert \bigcup_{a_i\in D}bp_i(M_t') \vert < \vert \bigcup_{a_i\in D}bp_i(M_t) \vert$} \Comment{Replace with $\vert ba(M_t')\cap D\vert < \vert ba(M_t)\cap D\vert$ to minimise the number of blocking $D$ agents instead}
            \State $M_t\gets M_t'$
        \EndIf
    \EndFor
    \State $M\gets M\cup M_t$
\EndFor

\State\Return{$M$}

\end{algorithmic}
\caption{Exact algorithm for {\sc Max-sri} problems with lists of length at most 2}
\label{alg:prefexact2smi}
\end{algorithm}

\begin{theorem}
\label{thm:prefexact2smi}
    Let $I=(A,\succ)$ be an {\sc sri} instance with $n$ agents and let $D\subseteq A$ be a set of deviators. If all preference lists of $I$ are of length at most 2, then Algorithm \ref{alg:prefexact2smi} computes an optimal solution to {\sc Deviator-Max-sri} and {\sc Deviator-BA-Max-sri} in $O(n^2)$ time.
\end{theorem}
\begin{proof}
    Let $M$ be the matching returned by Algorithm \ref{alg:prefexact2smi}. Again, notice that distinct connected components have -- by definition -- independent preference systems, so no two agents from distinct connected components can form a blocking pair together. The correctness of our treatment of paths involving $k$ agents follows easily: when $k$ is even, then the matches we add are the only maximum-cardinality matching of these agents, and any other way to match agents in such components leads to a matching of strictly smaller size. Furthermore, when $k$ is odd, the matchings $M_s= \{\{a_{r_{s+2x-1}},a_{r_{s+2x}}\}\;\vert\;1\leq x\leq \lfloor\frac{k}{2}\rfloor\}$ (where addition and subtraction are taken modulo $k$) that leave one $a_{r_s}$ agent (with an odd index $s$) unmatched are all the maximum-cardinality matchings of such components. Thus, we simply include the matching that minimises the number of blocking pairs involving $D$ agents (or the number of blocking $D$ agents). 
    
    It remains to argue why the way we deal with cycles is correct. Notice that each cycle $C=(a_{r_1}\;a_{r_2}\;a_{r_3}\dots a_{r_k})$ with an even number of agents $k$ admits only two different maximum-cardinality matchings of the agents in $C$. We consider both matchings and include the one that minimises the number of blocking pairs involving $D$ agents (or the number of blocking $D$ agents), so the correctness of this approach is easy to see. Notice, furthermore, that each such cycle with an odd number of agents $k$ admits $k$ different maximum-cardinality matchings, which are precisely maximum-cardinality matchings of the even-length paths resulting from the deletion of any one agent. We enumerate all such matchings and include the one that minimises the number of blocking pairs involving $D$ agents (or the number of blocking $D$ agents), so the correctness of this approach is easy to see too. Thus, the resulting matching we return is a maximum-cardinality matching of $I$ and minimises the total number of blocking pairs involving $D$ agents (or the total number of blocking $D$ agents).      

    Regarding time complexity, note that each component is dealt with exactly once. Before running the algorithm, we can create three different stacks for the three different types of components by exploring the instance in linear time. Now, we do have a nested for-loop for the case that deals with paths involving an odd number of agents and for the case that deals with cycles involving an odd number of agents. Asymptotically, each execution of the inner loop requires at most linear time in the length of the path or cycle, i.e., linear in the number of agents. Furthermore, we have at most a linear number of components in the number of agents. The quadratic upper bound follows as required.
\end{proof}

This positive result implies a tight tractability frontier: {\sc Deviator-Max-sri} and {\sc Deviator-BA-Max-sri} can both be solved very efficiently when preference lists are of length at most 2. However, the problems are {\sf NP-hard} even in the bipartite case and when preference lists are of length at most 3, as previously established in Corollaries \ref{cor:devmaxsmihard} and \ref{cor:devbahard}, respectively.

\section{Resolving Further Open Cases}
\label{sec:furtheropen}

When aggregating an overview of known and newly established complexity results in Table \ref{table:results} (to follow in Section \ref{sec:conclusion}), we identified two further cases that, to the best of our knowledge, have not yet been detailed explicitly in existing literature. Both concern the minimisation of blocking agents in {\sc sri} across the whole instance, i.e., the {\sc MinBA-AlmostStable-sri} problem, which is also the special case of {\sc Deviator-BA-sri} where $D=A$. It is known that {\sc MinBA-AlmostStable-sri} is {\sf NP-hard} when preference lists are of length 5 \cite{chen17}.

The first result shows that when preference lists are very short, then {\sc MinBA-AlmostStable-sri} is solvable in polynomial time. This follows immediately from our algorithm in Section \ref{sec:prefshort}.

\begin{theorem}
\label{thm:minbaSRIshort}
    {\sc MinBA-AlmostStable-sri} can be solved in $O(n)$ time when preference lists are of length at most 2 (where $n$ is the number of agents).
\end{theorem}
\begin{proof}
    We already noted above that {\sc MinBA-AlmostStable-sri} is the special case of the problem {\sc Deviator-BA-sri} where all agents are part of the deviator set. Therefore, due to Theorem \ref{thm:prefexact2sri}, Algorithm \ref{alg:prefexact2sri} solves {\sc MinBA-AlmostStable-sri} to optimality in, asymptotically, at most linear time in the number of agents.
\end{proof}

The second result shows that {\sc MinBA-AlmostStable-sri} remains intractable even when preference lists are complete.

\begin{theorem}
\label{thm:minbaSRIcomplete}
    {\sc MinBA-AlmostStable-sri} is {\sf NP-hard} even when all preference lists are complete.
\end{theorem}
\begin{proof}
    Suppose we could minimise the number of blocking agents in polynomial time. Then we could decide in polynomial time whether, for a given {\sc sri} instance $I$ with complete preference lists (which would be more appropriately denoted by {\sc src}) and a non-negative integer $k$, $I$ admits a matching with at most $k$ blocking agents. Note that this problem is {\sf NP-complete} when preference lists are permitted to be incomplete, as \citet{chen17} showed. We will denote the general problem (not requiring complete preference lists) by {\sc $k$-AlmostStable-sri} and the restricted problem requiring complete preference lists by {\sc $k$-AlmostStable-src}. Below, we give a simple reduction from the former problem to the latter.

    Let $(I=(A,\succ),k)$ be an instance of {\sc $k$-AlmostStable-sri} with $n$ agents. We construct an instance $(I'=(A',\succ'),k)$ of {\sc $k$-AlmostStable-src} as follows. Start with $I'=I$, i.e., every agent in $A$ is also in $A'$ and the preferences $\succ'$ are initially the same as in $\succ$. Then, for every agent $a_i\in A$, create $k$ additional dummy agents $a_i^1,a_i^2,\dots,a_i^{k}$ in $A'$. Furthermore, append these agents in increasing order of superscript to the end of $\succ_{a_i}'$, and let $\succ_{a_i^s}'$ (for all $1\leq s\leq k$) rank just $a_i$ to start with. After all dummy agents have been introduced and the modifications to $\succ'$ have been made, complete all preference lists of $\succ'$ by appending all previously unranked agents to the end of the preference list according to the global complete ranking $\succ_G$ on the agents in $A'$ such that $a_1\succ_G a_1^1\succ_G\dots\succ_G a_1^{k}\succ_G a_2\succ_G a_2^1\succ_G\dots\succ_G a_2^{k}\succ_G a_3\succ_G a_3^1\succ_G\dots \succ_G a_n^{k}$.

    To illustrate this, consider the following example with three agents $A=\{a_1,a_2,a_3\}$, and let $k=1$. Then, because $k=1$, we first create three dummy agents, one dummy agent per original agent in $A$), and let $A'=A\cup\{a_1^1,a_2^1,a_3^1\}$. The global ranking $\succ_G$ is $a_1\succ_G a_1^1 \succ_G a_2\succ_G a_2^1 \succ_G a_3\succ_G a_3^1$. Now suppose $a_1$'s preferences $\succ_{a_1}$ (in $\succ$) consist of just $a_3$. Then our final construction of $I'$ results in the preference list $\succ'_{a_1}$ (in $\succ'$) such that $a_3\succ'_{a_1} a_1^1 \succ'_{a_1} a_2\succ'_{a_1} a_2^1\succ'_{a_1} a_3^1$. If $a_1$'s preferences $\succ_{a_1}$ (in $\succ$) were $a_3\succ_{a_1}a_2$ instead, then our final construction of $I'$ would result in the preference list $\succ'_{a_1}$ (in $\succ'$) such that $a_3\succ'_{a_1} a_2\succ'_{a_1} a_1^1 \succ'_{a_1} a_2^1\succ'_{a_1} a_3^1$.
    
    Returning to the abstract setting, clearly $\vert A'\vert=(k+1)\vert A\vert$ and the preference lists have a total of $\Theta(\vert A'\vert^2)$ entries, so the construction is polynomial in our original instance $(I=(A,\succ),k)$. We will now establish a parameter-preserving correspondence between the two instances $I$ and $I'$ (with respect to parameter $k$) in the two claims below.

    \begin{claim}
    \label{claim:inccompba}
        If $I$ admits a matching $M$ such that $\vert ba^I(M)\vert \leq k$,\footnote{We use the notation $ ba^I(M)$ for clarity to indicate that we mean the set of blocking agents that $M$ admits with respect to instance $I$.} then $I'$ admits a matching $M'$ such that $\vert ba^{I'}(M')\vert \leq k$.
    \end{claim}
        \begin{proof}[Proof of Claim \ref{claim:inccompba}]
        \renewcommand{\qedsymbol}{$\blacksquare$}
        We prove this by constructing $M'$ explicitly from $M$, starting with $M'=M$. Now, for every unmatched agent $a_i\in A$, we can add $\{a_i,a_i^1\}$ to $M'$. Now, for the remaining unmatched dummy agents $a_j^s\in A'$, we match according to $\succ_G$: add the two best-ranked unmatched agents as a match to $M'$, then the third and fourth best-ranked unmatched agents, and so on. If the number of unmatched agents is odd, leave the last such agent unmatched.

        We will now prove that $bp^{I'}(M')\subseteq bp^I(M)$, in which case $ba^{I'}(M')\subseteq ba^I(M)$ and thus $\vert ba^{I'}(M')\vert\leq\vert ba^I(M)\vert \leq k$. Let $\{a_r,a_s\}\in bp^{I'}(M')$ (where either or both agents may be either dummy or non-dummy agents in $A'$). Then, by definition, both $a_s\succ_{a_r}' M'(a_r)$ and $a_r\succ_{a_s}' M'(a_s)$. Notice that no agent in $A$ can be in more blocking pairs in $M'$ than in $M$ because they either remain with the same partner or get their first choice in the extended preference list. Thus, if either $a_r\in A$ or $a_s\in A$, then both $a_r,a_s\in A$ and $\{a_r,a_s\}\in bp^{I}(M)$. Now suppose that both $a_r,a_s\in A'\setminus A$, i.e., both agents are dummy agents. Clearly, no dummy agent that is matched to their first choice is blocking, so both agents are matched according to $\succ_G$. Without loss of generality, suppose that $a_r\succ_G a_s$. Then, by construction of $M'$, either $a_r$ has a partner in $M'$ that is strictly better than $a_s$ (according to $\succ_G$ and thus also according to $a_s$' preferences) or $a_r$ is matched to $a_s$ in $M'$ (as we iteratively match the best two remaining unmatched agents according to $\succ_G$). Thus, $M'(a_r)\succ_G a_s$ or $M'(a_r)= a_s$, but, by construction of $\succ'$, then either $M'(a_r)\succ_{a_r}' a_s$ or $M'(a_r)= a_s$, contradicting that $a_s\succ_{a_r}' M'(a_r)$. Thus, no two dummy agents can form a blocking pair with respect to $I'$ in $M'$. Therefore, $bp^{I'}(M')=bp^I(M)$ and so $\vert ba^{I'}(M')\vert \leq k$ as required.
    \end{proof}

    \begin{claim}
    \label{claim:compincba}
        If $I'$ admits a matching $M'$ such that $\vert ba^I(M')\vert \leq k$, then $I$ admits a matching $M$ such that $\vert ba^{I'}(M)\vert \leq k$.
    \end{claim}
        \begin{proof}[Proof of Claim \ref{claim:compincba}]
        \renewcommand{\qedsymbol}{$\blacksquare$}
        Clearly, any agent $a_i\in A$ has a partner at least as good as $a_i^k$ in $M'$ according to $\succ_{a_i}'$, as otherwise every $a_i^s$ (for $1\leq s\leq k$) would be blocking with $a_i$, in which case $\vert ba^{I'}(M')\vert \geq k+1$ (all $k$ dummy agents block and $a_i$ blocks with them), contradicting our assumption. Thus, we can simply construct a matching $M$ of $I$ in which, for every agent $a_i\in A$, if $M'(a_i)\in A$, then $M(a_i)=M'(a_i)$, and $M(a_i)=a_i$ (i.e., we leave $a_i$ unmatched) otherwise. It is easy to see that $\vert ba^I(M)\vert=\vert ba^{I'}(M')\cap A\vert \leq \vert ba^{I'}(M')\vert\leq k$.
    \end{proof}

    The stated result follows immediately.
\end{proof}

\section{Conclusion}
\label{sec:conclusion}

In this paper, we considered a simple premise: some agents can initiate blocking pair-based deviations, and some agents cannot. We aimed to answer the following question from a computational complexity point of view: Can we find desirable matchings in which no deviator is incentivised to deviate? The short answer is yes, but only under strong assumptions on the input. Our key insights are summarised in Table \ref{table:results} and reveal a sharp contrast between strong intractability and tractable special cases. 

\begin{table*}[!tbh]
    \centering
    \footnotesize
    \setlength{\aboverulesep}{0pt}
    \setlength{\belowrulesep}{0pt}
    \renewcommand{\arraystretch}{1.1}
    \renewcommand\cellgape{\Gape[2pt]}  
    \caption{Overview of complexity results. Here, $d$ denotes the maximum preference list length (also denoted by $d_{\max}$ in other parts of the paper), and CL indicates complete preference lists. Our contributions are coloured in \textcolor{blue}{blue}. The following symbols are used: ${\top}$ (result holds trivially), $\dagger$ (positive result only established for bipartite {\sc smi} instances), $\ddagger$ (negative result holds even for bipartite {\sc smi} instances), $^*$ (we newly introduced these problems).}
    \begin{tabular}{c | c | c}
        \toprule
        & {\sc sri} & {\sc Max-sri}  \\
        \midrule
        \makecell{{\sc MinBP}\\($\kappa=\min \vert bp(M)\vert$)} 
                   & \makecell{{\sf P} ($d\leq 2$) \cite{biro12}, {\sf XP} ($\kappa$) \cite{abraham06} \\ 
                   {\sf NP-h} ($d\leq 3$ and CL) \cite{abraham06,biro12}, {\sf W[1]-h} ($\kappa$) \cite{chen17}}
                   & \makecell{{\sf P}$^\dagger$ ($d\leq 2$ and CL) \cite{biro_sm_10}, {\sf XP}$^\dagger$ ($\kappa$) \cite{biro_sm_10}\\
                   {\sf NP-h}$^\ddagger$ ($d\leq 3$) \cite{biro_sm_10}} \\
        \midrule
        \makecell{{\sc MinBA}\\($\kappa=\min \vert ba(M)\vert$)} 
                   & \makecell{\textcolor{blue}{{\sf P} ($d\leq 2$) [T\ref{thm:minbaSRIshort}]}, {\sf XP} ($\kappa$) \cite{chen17}\\
                   {\sf NP-h} ($d\leq 5$ \cite{chen17} \textcolor{blue}{and CL [T\ref{thm:minbaSRIcomplete}]}), {\sf W[1]-h} ($\kappa$) \cite{chen17}}
                   & \makecell{{\sf P}$^\dagger$ ($d\leq 2$ and CL) \cite{biro_sm_10}, {\sf XP}$^\dagger$ ($\kappa$) \cite{biro_sm_10}\\
                   {\sf NP-h}$^\ddagger$ ($d\leq 3$) \cite{biro_sm_10}} \\
        \midrule
        \makecell{{\sc Minimax}\\($\kappa=\min\max \vert bp_i(M)\vert$)} 
                   & {\makecell{{\sf P} ($d\leq 2$) \cite{glitznermanloveminmax}\\
                   {\sf NP-h} ($\kappa=1$, $d\leq 10$ and CL) \cite{glitznermanloveminmax}}}
                   & {\makecell{{\sf P}$^\dagger$ ($d\leq 2$ and CL) [\cite{glitznermanloveminmax},$\top$]\\
                   {\sf NP-h}$^\ddagger$ ($\kappa=1$, $d\leq 3$) \cite{glitznermanloveminmax}}} \\
        \midrule
        \makecell{{\sc Deviator}$^*$\\($\kappa_1=\vert D\vert$,\\$\kappa_2=\min \vert \bigcup_{a_i\in D}bp_i(M)\vert$)} 
                   & \textcolor{blue}{\makecell{{\sf P} ($d\leq 2$) [T\ref{thm:prefexact2sri}]\\ {\sf FPT} ($(d,\kappa_1)$) [C\ref{cor:psrifpt}]\\ {\sf XP} ($(\kappa_1,\kappa_2)$) [C\ref{cor:psrifpt}]\\
                   {\sf NP-h} ($\kappa_2=0$, $d\leq 5$ and CL) [T\ref{thm:prefsrihard}]}}
                   & \textcolor{blue}{\makecell{{\sf P} ($d\leq 2$) [T\ref{thm:prefexact2smi}]\\ {\sf FPT} ($(d,\kappa_1)$) [C\ref{cor:preffptxp}]\\ {\sf XP} ($(\kappa_1,\kappa_2)$) [C\ref{cor:preffptxp}]\\
                   {\sf NP-h}$^\ddagger$ ($\kappa_2=0$, $d\leq 3$) [T\ref{thm:preferentialbip}]}} \\
        \midrule
        \makecell{{\sc Deviator-BA}$^*$\\($\kappa_1=\vert D\vert$,\\$\kappa_2=\min \vert ba(M)\cap D\vert$)} 
                   & \textcolor{blue}{\makecell{{\sf P} ($d\leq 2$) [T\ref{thm:prefexact2sri}]\\ {\sf FPT} ($(d,\kappa_1)$) [C\ref{cor:prefba}]\\ {\sf XP} ($(\kappa_1,\kappa_2)$) [C\ref{cor:prefba}]\\
                   {\sf NP-h} ($\kappa_2=0$, $d\leq 5$ and CL) [T\ref{thm:prefBAhard}]}}
                   & \textcolor{blue}{\makecell{{\sf P} ($d\leq 2$) [T\ref{thm:prefexact2smi}]\\ {\sf FPT} ($(d,\kappa_1)$) [C\ref{cor:prefba}]\\ {\sf XP} ($(\kappa_1,\kappa_2)$) [C\ref{cor:prefba}]\\
                   {\sf NP-h}$^\ddagger$ ($\kappa_2=0$, $d\leq 3$) [T\ref{thm:prefBAhard}]}} \\
        \bottomrule
    \end{tabular}
    \label{table:results}
\end{table*}

We conclude with a number of open questions. For instance, on the parametrised complexity front, are the problems {\sc Deviator-sri} and {\sc Deviator-Max-sri} in {\sf FPT} with respect to just $\vert D\vert$, rather than the combined parameter $(\vert D\vert, d_{\max})$? We conjecture that the answer is no, as it might be the case that $\vert D\vert\ll d_{\max}$, but our current reductions rely on the fact that $D$ is sufficiently large. It would also be interesting to study approximation algorithms in the setting of deviator stability, although we highlight that strong inapproximability results for almost-stable matchings are known \cite{abraham06,biro_sm_10,chen2025fptapproximabilitystablematchingproblems} and that we proved that even deciding whether there exists a matching where no deviator is blocking is {\sf NP-complete}. Thus, it might be interesting to study algorithms that give guarantees with respect to the number of deviators that are in no or in few blocking pairs, while ignoring the remaining deviators. Lastly, on a more practical front, it would be interesting to design, deploy and study real-world multi-agent systems using deviator-stable matchings -- how do we accurately identify deviators in practice, and are there sufficiently few to render our parametrised algorithms efficient?

\paragraph{Statements and Declarations}
Frederik Glitzner is supported by a Minerva Scholarship from the School of Computing Science, University of Glasgow. 

\vspace{-1em}

\paragraph{Acknowledgements}
The authors would like to thank the anonymous 
reviewers for helpful comments and suggestions on an earlier version of this paper. In particular, two comments led to Remarks \ref{rm:otherreduction} and \ref{rm:freeedges}.

\bibliographystyle{reference_format}
\bibliography{papers}

\clearpage
\appendix

\section{Omitted Full Proofs of {\sf FPT} Algorithms}

\subsection{{\sc $k$-Deviator-sri} {\sf FPT} Algorithm Proof}
\label{apx:kdevsri}

\preferentialfpt*
\begin{proof}
    Recall that we aim to find a matching $M_P$ of $I$ such that the agents in $D$ are collectively involved in at most $k$ blocking pairs in $M_P$. Again, we will denote such blocking pairs by $bp(M_P)\vert_{D}$. Clearly, every agent in $a_i\in D$ is either matched to one of the at most $d_{\max}$ many agents on their preference list or remains unmatched. Thus, there are $(d_{\max}+1)^{\vert D\vert}=O(d_{\max}^{\vert D\vert})$ possible combinations of choices of partner, including the possibility of being unmatched, for the agents in $D$. We can immediately discard any combinations of choices that are not matchings, and we denote the remaining set of candidate matchings by $\mathcal{M}_D$. Next, we consider each candidate matching sequentially and aim to extend it to a solution to our problem. Let us fix a candidate matching $M_D\in \mathcal{M}_D$.

    Now, we consider each possible set of blocking pairs involving some agent in $D$ of size at most $k$. There are clearly at most $\sum_{1\leq r\leq k}\binom{\vert D\vert d_{\max}}{r}$ many such sets. We can apply the loose upper bound $\binom{\vert D\vert d_{\max}}{r}\leq (\vert D\vert d_{\max})^r$ and conclude that $\sum_{1\leq r\leq k} (\vert D\vert d_{\max})^r=\frac{(\vert D\vert d_{\max})^{k+1}-1}{\vert D\vert d_{\max}-1}$ by the geometric series. Furthermore, the bound $\frac{(\vert D\vert d_{\max})^{k+1}-1}{\vert D\vert d_{\max}-1}\leq \frac{(\vert D\vert d_{\max})^{k+1}}{\vert D\vert d_{\max}-1}\leq 2(\vert D\vert d_{\max})^{k}=O((\vert D\vert d_{\max})^{k})$ applies. Let us fix a candidate set of blocking pairs $B$. Clearly, no pair of agents can be simultaneously matched and blocking, so if $M_D\cap B\neq \varnothing$, then we reject this configuration $(M_D,B)$.

    Then, for every agent $a_i\in D$ and every agent $a_r\in A$ such that $a_r\succ_i M_D(a_i)$ and $\{a_i,a_r\}\notin B$, agent $a_r$ must end up with a partner better than $a_i$ (according to $\succ_r$). Thus, we truncate $\succ_r$, discarding every agent worse than and including $a_i$ (if multiple agents in $D$ satisfy this criterion with respect to $a_r$, we truncate at the best-ranked such agent according to $\succ_r$). We must keep track of these agents $a_r$ to later verify that they do indeed have a partner, and denote the set of all such agents $a_r$ by $Q$.

    If, after carrying out these truncations, some agent $a_i\in D$ has a partner in $M_D$ that is no longer in their preference list (or vice versa), we can reject this configuration $(M_D,B)$. Otherwise, we need to look for a matching among the agents not yet matched in $M_D$ such that every agent in $Q$ is matched. We can do this as follows: we construct a maximum-weight matching instance $(G,w)$ (where $G$ is a graph and $w$ is a weight function from the set of edges to non-negative integers) consisting of all agents that are at most two agents away from some agent in $D$ and not yet matched in $M_D$ as vertices\footnote{Here, and below, the construction differs from that in Theorem \ref{thm:preferentialbipfpt}: we do not need to consider the whole set of agents because we are no longer requiring a maximum-cardinality matching overall. This also simplifies the weight function.} and all acceptable pairs among these agents as edges. Furthermore, we construct $w$ as follows: for every acceptable pair of agents $a_r,a_s$ present in the instance, we let $w(\{a_r,a_s\})=\vert\{a_r,a_s\}\cap Q\vert$, i.e., we assign a weight of 1 for every endpoint in $Q$. 
    
    We can construct $(G,w)$ in linear time in the graph we create, assuming that preference lists can be explored efficiently given an agent. For example, given that the set of agents $D$ is given explicitly, the preference lists of each such agent could be implemented in a linked-list-style representation, which would allow the exploration of all relevant matchings without necessarily exploring the whole instance. Note that by the preference list length, there are at most $\vert D\vert d_{\max}^2$ many agents that are at most two agents away from some agent in $D$, i.e, $G$ contains at most $\vert D\vert d_{\max}^2$ vertices, and every such agent has a preference list of length at most $d_{\max}$, so $G$ has at most $\vert D\vert d_{\max}^3$ edges. Thus, we can create $(G,w)$ in $O(\vert D\vert d_{\max}^3)$ time. Now we find a maximum-weight matching $M_{mw}$ with weight $w(M_{mw})$ for (not necessarily bipartite) instance $(G,w)$, e.g., using the algorithm by \citet{huangkavitha}, which runs in $O(\sqrt{n_*}m_*N\log(n_*^2/m_*)/\log(n_*))$ time, where $n_*$ is the number of vertices in the graph, $m_*$ is the number of edges in the graph and $N$ is the maximum edge weight. Thus, given that $n_*\leq \vert D\vert d_{\max}^2$, $m_*\leq \vert D\vert d_{\max}^3$ and $N\leq 2$, we have that $\sqrt{n_*}m_*N\leq2(\vert D\vert d_{\max}^2)^{1/2}\vert D\vert d_{\max}^3=O(\vert D\vert^{3/2}d_{\max}^4)$ and $\log(n_*^2/m_*)/\log(n_*)\leq\log((\vert D\vert d_{\max}^2)^2/\vert D\vert d_{\max}^3)/\log(\vert D\vert d_{\max}^2)=\log(\vert D\vert d_{\max})/\log(\vert D\vert d_{\max}^2)=O(1)$. Thus, this algorithm runs in $O(\vert D\vert^{3/2}d_{\max}^4)$ time.

    Clearly, by construction of $w$, $M_{mw}$ matches the maximum possible number of agents in $Q$. Thus, $M_{mw}$ is a matching of $(G,w)$ in which every agent in $Q$ is matched if and only if such a matching exists. Therefore, if there exists an agent $a_r\in Q$ that is not matched in $M_{mw}$, we can reject this configuration $(M_D,B)$. Otherwise, by construction, $M_D$ and $M_{mw}$ are disjoint, so we can accept $M=M_D\cup M_{mw}$ as our solution. 

    Of course, if we accept $M$, then $M$ is a matching of $I$ and, by construction, $\vert bp(M)\vert_{D}\vert \leq \vert B\vert\leq k$. Thus, $M$ is a solution to {\sc $k$-Deviator-sri}. Otherwise, if, for every configuration $(M_D,B)$ and an associated maximum-weight matching $M_{mw}$, it is the case that not every agent in $Q$ is matched, then it must be true that $(I,D)$ is a no-instance to {\sc $k$-Deviator-sri}, i.e., that there exists no solution to the problem for this instance. To see this, notice that if $(I=(A,\succ), D)$ has a solution $M_P$, then $M_P$ can be divided into a part $M_1$ in which every pair contains at least one agent from $D$, and a part $M_2$ which does not contain any pair with an agent from $D$. Clearly, $M_1$ satisfies the requirements of the candidate matchings $\mathcal{M}_D$, and $M_2$ must be a matching among the remaining agents. Let us fix this $M_1\in \mathcal{M}_D$. By being a solution, it must furthermore be true that, for every agent $a_i\in D$ and every agent $a_r\in A$ such that $a_r\succ_i M_P(a_i)$, either $\{a_i,a_r\}\in bp(M_P)\vert_{D}$ or ($\{a_i,a_r\}\notin bp(M_P)\vert_{D}$ and $M_P(a_r)\succ_r a_i$). We keep track of agents satisfying the latter case in set $Q$. Note that it must be the case that $\vert bp(M_P)\vert_{D}\vert=\vert bp(M_1\cup M_2)\vert_{D}\vert \leq k$, so we would have considered this set $bp(M_P)\vert_{D}$ as a candidate set $B$ in our execution. Furthermore, by construction, $M_P(a_r)\succ_r a_i$ for every $a_r\in Q$. Hence, there exists a maximum-weight matching $M_{mw}$ of $(G,w)$ such that $M_{mw}\subseteq M_2$ and $M_{mw}(a_r)=M_P(a_r)$ for every $a_r\in Q$. Thus, our algorithm would have identified solution $M_1\cup M_{mw}$, or a different solution, to the problem.

    Overall, we consider $O(d_{\max}^{\vert D\vert})$ candidate matchings and, for every candidate matching, we consider $O((\vert D\vert d_{\max})^{k})$ sets of blocking pairs, leading to $O(d_{\max}^{\vert D\vert}(\vert D\vert d_{\max})^{k})=O(\vert D\vert^{k}d_{\max}^{\vert D\vert+k})$ candidate configurations. For each configuration, we construct a maximum-weight matching instance in $O(\vert D\vert d_{\max}^3)$ time and subsequently compute a maximum-weight matching in $O(\vert D\vert^{3/2}d_{\max}^4)$ time. Hence, we arrive at a final complexity of $O(\vert D\vert^{k}d_{\max}^{\vert D\vert+k}\vert D\vert^{3/2}d_{\max}^4)=O(\vert D\vert^{k+3/2}d_{\max}^{\vert D\vert+k+4})$.
\end{proof}

\subsection{{\sc $k$-Deviator-BA-sri} {\sf FPT} Algorithm Proof}
\label{apx:kdevbasri}

\preferentialBAfpt*
\begin{proof}
    The high-level approach remains the same as in the proof of Theorem \ref{thm:preferentialfpt}. We consider $(d_{\max}+1)^{\vert D\vert}=O(d_{\max}^{\vert D\vert})$ possible combinations of choices of partner, including the possibility of being unmatched, for the agents in $D$. We can immediately discard any combinations of choices that are not matchings, and we denote the remaining set of candidate matchings by $\mathcal{M}_D$. Next, we consider each candidate matching sequentially and aim to extend it to a solution to our problem. Let us fix a candidate matching $M_D\in \mathcal{M}_D$.

    Now, we consider each possible set of blocking agents in $D$ of size at most $k$.\footnote{Here, and below, the approach differs slightly from that in Theorem \ref{thm:preferentialfpt}: we do not need to consider the whole set of blocking pairs but rather sets of blocking agents.} There are clearly at most $\sum_{1\leq r\leq k}\binom{\vert D\vert}{r}$ many such sets. We can apply the loose upper bound $\binom{\vert D\vert}{r}\leq \vert D\vert^r$ and conclude that $\sum_{1\leq r\leq k} \vert D\vert^r=\frac{\vert D\vert^{k+1}-1}{\vert D\vert-1}$ by the geometric series. Furthermore, the bound $\frac{\vert D\vert^{k+1}-1}{\vert D\vert -1}\leq \frac{\vert D\vert^{k+1}}{\vert D\vert-1}\leq 2\vert D\vert^{k}=O(\vert D\vert^{k})$ applies. Let us fix a candidate set of blocking agents $B$. 
    
    Then, for every agent $a_i\in D\setminus B$ and every agent $a_r\in A$ such that $a_r\succ_i M_D(a_i)$, agent $a_r$ must end up with a partner better than $a_i$ (according to $\succ_r$). Thus, we truncate $\succ_r$, discarding every agent worse than and including $a_i$ (if multiple agents in $D\setminus B$ satisfy this criterion with respect to $a_r$, we truncate at the best-ranked such agent according to $\succ_r$). We must keep track of these agents $a_r$ to later verify that they do indeed have a partner, and denote the set of all such agents $a_r$ by $Q$.

    From here, the approach is the same again as in the proof of Theorem \ref{thm:preferentialfpt}. We find a maximum-weight matching $M_{mw}$ that maximises matches involving agents in $Q$ in $O(\vert D\vert^{3/2}d_{\max}^4)$ time. If there exists an agent $a_r\in Q$ that is not matched in $M_{mw}$, we can reject this configuration $(M_D,B)$. Otherwise, by construction, $M_D$ and $M_{mw}$ are disjoint, so we can accept $M=M_D\cup M_{mw}$ as our solution. 

    Of course, if we accept $M$, then $M$ is a matching of $I$ and, by construction, $\vert ba(M)\cap D\vert \leq \vert B\vert\leq k$. Thus, $M$ is a solution to {\sc $k$-Deviator-BA-sri}. Otherwise, if, for every configuration $(M_D,B)$ and an associated maximum-weight matching $M_{mw}$, it is the case that not every agent in $Q$ is matched, then it must be true that $(I,D)$ is a no-instance to {\sc $k$-Deviator-BA-sri}, i.e., that there exists no solution to the problem for this instance. To see this, notice that if $(I=(A,\succ), D)$ has a solution $M_P$, then $M_P$ can be divided into a part $M_1$ in which every pair contains at least one agent from $D$, and a part $M_2$ which does not contain any pair with an agent from $D$. Clearly, $M_1$ satisfies the requirements of the candidate matchings $\mathcal{M}_D$, and $M_2$ must be a matching among the remaining agents. Let us fix this $M_1\in \mathcal{M}_D$. By being a solution, it must furthermore be true that, for every agent $a_i\in D$ and every agent $a_r\in A$ such that $a_r\succ_i M_P(a_i)$, either $a_i\in ba(M_P)\cap D$ or ($a_i\notin ba(M_P)\cap D$ and $M_P(a_r)\succ_r a_i$). We keep track of agents satisfying the latter case in the set $Q$. Note that it must be the case that $\vert ba(M_P)\cap D\vert \leq k$, so we would have considered this set $ba(M_P)\cap D$ as a candidate set $B$ in our execution. Furthermore, by construction, $M_P(a_r)\succ_r a_i$ for every $a_r\in Q$. Hence, there exists a maximum-weight matching $M_{mw}$ of $(G,w)$ such that $M_{mw}\subseteq M_2$ and $M_{mw}(a_r)=M_P(a_r)$ for every $a_r\in Q$. Thus, our algorithm would have identified solution $M_1\cup M_{mw}$, or a different solution, to the problem.

    Overall, we consider $O(d_{\max}^{\vert D\vert})$ candidate matchings and, for every candidate matching, we consider $O(\vert D\vert^{k})$ sets of blocking agents, leading to $O(d_{\max}^{\vert D\vert}\vert D\vert ^{k})$ candidate configurations. For each configuration, we construct a maximum-weight matching instance in $O(\vert D\vert d_{\max}^3)$ time and subsequently compute a maximum-weight matching in $O(\vert D\vert^{3/2}d_{\max}^4)$ time. Hence, we arrive at a final complexity of $O(\vert D\vert^{k}d_{\max}^{\vert D\vert}\vert D\vert^{3/2}d_{\max}^4)=O(\vert D\vert^{k+3/2}d_{\max}^{\vert D\vert+4})$.
\end{proof}

\subsection{{\sc $k$-Deviator-BA-Max-sri} {\sf FPT} Algorithm Proof}
\label{apx:kdevbamaxsri}

\preferentialBAbipfpt*
\begin{proof}
    Again, the high-level approach remains the same as in the proof of Theorem \ref{thm:preferentialBAfpt}. We consider at most $(d_{\max}+1)^{\vert D\vert}=O(d_{\max}^{\vert D\vert})$ possible candidate matchings $\mathcal{M}_D$ and at most $O(\vert D\vert^k)$ many possible sets of blocking agents in $D$. Similarly, for each combination of candidate matching $M_D$ and candidate set of blocking agents $B$, we create the set $Q$ of agents that must be matched in a corresponding solution as in the proof of Theorem \ref{thm:preferentialBAfpt}. 
    
    From hereon, the proof mirrors that of Theorem \ref{thm:preferentialbipfpt}. We find a maximum-weight matching $M_{mw}$ that maximises the number of overall matches and, subject to this, maximises the number of matched agents in $Q$. This can be done in  $O(d_{\max}n^{5/2})$ time. If there exists an agent $a_r\in Q$ that is not matched in $M_{mw}$, we can reject this configuration $(M_D,B)$. Otherwise, by construction, $M_D$ and $M_{mw}$ are disjoint, so if $\vert M_D\cup M_{mw}\vert = \vert M_S\vert$, then we can accept $M=M_D\cup M_{mw}$ as our solution. Otherwise, $M_D$ does not extend to a solution with respect to $B$, and we can reject this configuration $(M_D,B)$.

    Of course, if we accept $M$, then $M$ is a maximum-cardinality matching of $I$. Furthermore, by construction, $\vert ba(M)\cap D\vert \leq \vert B\vert\leq k$. Thus, $M$ is a solution to {\sc $k$-Deviator-BA-Max-sri}. Otherwise, if, for every configuration $(M_D,B)$ and an associated maximum-weight matching $M_{mw}$, it is the case that either $\vert M_D\cup M_{mw}\vert<\vert M_S\vert$ or not every agent in $Q$ is matched, then it must be true that $(I,D)$ is a no-instance to {\sc $k$-Deviator-BA-Max-sri}, i.e., that there exists no solution to the problem for this instance. To see this, notice that if $(I=(A,\succ), D)$ has a solution $M_P$, then $M_P$ can be divided into a part $M_1$ in which every pair contains at least one agent from $D$, and a part $M_2$ which does not contain any pair with an agent from $D$. Clearly, $M_1$ satisfies the requirements of the candidate matchings $\mathcal{M}_D$, and $M_2$ must be a maximum-cardinality matching among the remaining agents (otherwise $M_P$ would not be a maximum-cardinality matching). Let us fix this $M_1\in \mathcal{M}_D$. 

    By being a solution, it must furthermore be true that, for every agent $a_i\in D$ and every agent $a_r\in A$ such that $a_r\succ_i M_P(a_i)$, either $a_i\in ba(M_P)\cap D$ or ($a_i\notin ba(M_P)\cap D$ and $M_P(a_r)\succ_r a_i$). We keep track of agents satisfying the latter case in the set $Q$. Note that it must be the case that $\vert ba(M_P)\cap D\vert \leq k$, so we would have considered this set $ba(M_P)\cap D$ as a candidate set $B$ in our execution. Furthermore, $M_2$ must be a maximum-cardinality matching of our maximum-weight matching instance $(G,w)$ (for some configuration $(M_D,B)$) that matches all such agents $Q$, and, by the key property we stated about maximum-weight matchings in Theorem \ref{thm:preferentialbipfpt}, $M_2$ must be a maximum-weight matching of $(G,w)$, too. Thus, our algorithm would have identified this solution $M_1\cup M_{mw}$, or a different solution, to the problem.

    Putting the complexity analysis together, we consider $O(d_{\max}^{\vert D\vert})$ candidate matchings and, for every candidate matching, we consider $O(\vert D\vert^{k})$ sets of blocking agents, leading to $O(d_{\max}^{\vert D\vert}\vert D\vert^{k}=O(\vert D\vert^{k}d_{\max}^{\vert D\vert})$ candidate configurations. We construct a maximum-weight matching instance in $O(d_{\max}n)$ time for each candidate configuration and, subsequently, compute a maximum-weight matching in $O(d_{\max}n^{5/2})$ time. Hence, we arrive at a final complexity of $O(\vert D\vert^{k}d_{\max}^{\vert D\vert+1}n^{5/2})$.
\end{proof}

\end{document}